\RequirePackage{ifthen}
\RequirePackage{ifpdf}
\newcommand{\driverOption}{}
\ifthenelse{\boolean{pdf}}{
  \renewcommand{\driverOption}{pdftex}
} { % else
  \renewcommand{\driverOption}{dvips}
}

\documentclass[reqno, 11pt, letterpaper, commented, oneside, \driverOption]{amsart}

\newboolean{isCommented}
\usepackage{geometry}
\usepackage{bbm}
\usepackage[final]{graphicx}
\usepackage{upref}
\usepackage{enumitem}
\usepackage{latexsym}
\usepackage{amssymb}
\usepackage[ansinew]{inputenc}
\usepackage[T1]{fontenc}
\usepackage[ruled,vlined,norelsize]{algorithm2e}  % norelsize needed to compile on arxiv
\usepackage{mathtools}
\usepackage{mycommands}
\usepackage{soul}

\ifthenelse{\boolean{pdf}}{
  \usepackage[final]{ps4pdf}
} { % else
  \usepackage[inactive]{ps4pdf}
}
\PSforPDF{\usepackage{psfrag}}

\newcommand{\hyperrefDriverOption}{}
\ifthenelse{\boolean{pdf}}{
	\renewcommand{\hyperrefDriverOption}{pdftex}
} { % else
	\renewcommand{\hyperrefDriverOption}{hypertex}
}
\usepackage[\hyperrefDriverOption,
  colorlinks = false, 
  pdfauthor={Torsten Mütze, Jerri Nummenpalo},
  pdftitle={Efficient computation of middle levels Gray codes}]
  {hyperref}
    
\ifthenelse{\boolean{isCommented}} {
	\newcommand{\TM}[1]{\marginpar{\parbox{4cm}{{\small {\bf TM:} #1}}}}
	\newcommand{\JN}[1]{\marginpar{\parbox{4cm}{{\small {\bf JN:} #1}}}}
} { % else
	\newcommand{\TM}[1]{}
	\newcommand{\JN}[1]{}
}
\newtheorem{theorem}{Theorem}
\newtheorem{lemma}[theorem]{Lemma}

\theoremstyle{definition}

\theoremstyle{remark}
\newtheorem{remark}[theorem]{Remark}

\long\def\symbolfootnote[#1]#2{\begingroup
\def\thefootnote{\fnsymbol{footnote}}\footnote[#1]{#2}\endgroup}

\ifthenelse{\boolean{isCommented}} {
	\geometry{%showframe,
	  hmargin={25mm, 50mm},
	  marginparwidth=40mm, 
	  vmargin={25mm, 25mm},
	  headsep=10mm,
	  headheight=5mm,
	  footskip=10mm
	}
} { % else
	\geometry{%showframe,
	  hmargin={25mm, 25mm}, 
	  vmargin={25mm, 25mm},
	  headsep=10mm,
	  headheight=5mm,
	  footskip=10mm
	}
}

\begin{document}

\begin{center}

\renewcommand{\thefootnote}{\fnsymbol{footnote}}

\LARGE Efficient computation of middle levels Gray codes\footnote{An extended abstract of this work has appeared in the proceedings of the European Symposium on Algorithms (ESA) 2015.}
\vspace{2mm}

\small

\begingroup
\begin{tabular}{l@{\hspace{2em}}l}
  \Large Torsten Mütze\footnotemark[2] & \Large Jerri Nummenpalo \\[2mm]
  School of Mathematics & Department of Computer Science \\
  Georgia Institute of Technology & ETH Zürich \\
  30332 Atlanta GA, USA & 8092 Zürich, Switzerland \vspace{.05mm} \\
  {\small {\tt muetze@math.gatech.edu}} & {\small {\tt njerri@inf.ethz.ch}}
\end{tabular}%
\footnotetext[2]{The author was supported by a fellowship of the Swiss National Science Foundation. 
This work was initiated when the author was at the Department of Computer Science at ETH Zürich.}%
\endgroup

\vspace{5mm}

\small

\begin{minipage}{0.8\linewidth}
\textsc{Abstract.}
For any integer $n\geq 1$ a \emph{middle levels Gray code} is a cyclic listing of all bitstrings of length $2n+1$ that have either $n$ or $n+1$ entries equal to 1 such that any two consecutive bitstrings in the list differ in exactly one bit. 
The question whether such a Gray code exists for every $n\geq 1$ has been the subject of intensive research during the last 30 years, and has been answered affirmatively only recently [T.~Mütze. Proof of the middle levels conjecture. \textit{Proc. London Math. Soc.}, 112(4):677--713, 2016]. 
In this work we provide the first efficient algorithm to compute a middle levels Gray code. 
For a given bitstring, our algorithm computes the next $\ell$ bitstrings in the Gray code in time $\cO(n\ell(1+\frac{n}{\ell}))$, which is $\cO(n)$ on average per bitstring provided that $\ell=\Omega(n)$.
\end{minipage}

\vspace{2mm}

\begin{minipage}{0.8\linewidth}
\end{minipage}

\vspace{2mm}

\end{center}

\vspace{5mm}

\section{Introduction}
\label{sec:introduction}

Efficiently generating all objects in a particular combinatorial class (e.g.\ permutations, combinations, partitions or trees) in such a way that each object is generated exactly once is one of the oldest and most fundamental problems in the area of combinatorial algorithms, and such generation algorithms are used as core building blocks in a wide range of practical applications.
The survey \cite{Savage:1997} lists numerous references.
A classical example is the so-called binary \emph{Gray code} \cite{gray:patent}, which lists all $2^n$ bitstrings of length $n$ such that any two consecutive bitstrings differ in exactly one bit. 
A straightforward implemention of this algorithm takes time $\cO(n)$ to compute from a given bitstring the next one in the list (see Algorithm~G in \cite[Section 7.2.1.1]{knuth}), which can be improved to $\cO(1)$ \cite{MR0366085,MR0424386} (see Algorithm~L in \cite[Section 7.2.1.1]{knuth}).
The space requirement of both algorithms is $\cO(n)$.
Similar minimum-change generation algorithms have been developed for various other combinatorial classes. 
We exemplarily cite four examples from the excellent survey \cite{Savage:1997} on this topic: (1) listing all permutations of $\{1,2,\ldots,n\}$ so that consecutive permutations differ only by the swap of one pair of adjacent elements \cite{MR0159764,Trotter:1962} (see also \cite{Sedgewick:1977}),
(2) listing all $k$-element subsets of an $n$-element set such that consecutive sets differ only by exchanging one element \cite{MR0366085,MR0424386,EadesEtAl:1984,EadesMcKay:1984,Ruskey:1988},
(3) listing all binary trees with $n$ vertices so that consecutive trees differ only by one rotation operation \cite{Lucas:1987,LucasEtAl:1993},
(4) listing all spanning trees of a graph such that consecutive trees differ only by exchanging one edge \cite{Cummins:1966,holzmannHarary:1972}.

Coming back to Gray codes, we say that a bitstring of length $n$ has \emph{weight} $k$, if it has exactly $k$ entries equal to 1 and $n-k$ entries equal to 0.
Furthermore, for any integer $n\geq 1$ we define a \emph{middle levels Gray code} as a cyclic listing of all bitstrings of length $2n+1$ that have weight $n$ or $n+1$ such that any two consecutive bitstrings in the list differ in exactly one bit.
The name `middle levels' becomes clear when considering the relevant bitstrings as subsets in the Hasse diagram of the subset inclusion lattice.
Clearly, a middle levels Gray code has to visit $N:=\binom{2n+1}{n}+\binom{2n+1}{n+1}=2\binom{2n+1}{n}=2^{\Theta(n)}$ many bitstrings in total, and the weight of the bitstrings will alternate between $n$ and $n+1$ in every step.
The existence of a middle levels Gray code for every value of $n$ is asserted by the infamous \emph{middle levels conjecture}, which originated probably with Havel~\cite{MR737021} and Buck and Wiedemann~\cite{MR737262}, but has also been attributed to Dejter, Erd{\H{o}}s, Trotter~\cite{MR962224} and various others. 
It also appears as Exercise~56 in Knuth's book~\cite[Section~7.2.1.3]{knuth} and in the popular book \cite{MR2858033}.
This conjecture also became known as \emph{revolving door conjecture} for the following reason: Imagine a set of $2n+1$ people, split into two subsets of size $n$ and $n+1$ that are separated by a revolving door. 
The conjecture asks whether it is possible to move in each step one person from the larger group through the revolving door to the other side, such that every way to split the people into two subsets of size $n$ and $n+1$ is encountered exactly once.
It may come as a surprise that establishing the existence of a middle levels Gray code appears to be a difficult problem, given that by item~(2) above one can easily come up with a listing of all bitstrings of length $2n+1$ with weight exactly $n$ such that any two consecutive bitstrings differ in \emph{two} bits.

The middle levels conjecture has attracted considerable attention over the last 30 years \cite{savage:93,MR1350586,MR1329390,MR2046083,MR962223,MR962224,MR1268348,horakEtAl:05,Gregor20102448}. 
Until recently, middle levels Gray codes had only been found with brute-force computer searches for $n\leq 19$ \cite{MR1745213,MR2548541,shimada-amano}.
For $n=19$ this Gray code already consists of $N=137.846.528.820$ bitstrings.
A complete proof of the conjecture has only been announced very recently.

\begin{theorem}[\cite{MR3483129}]
\label{thm:middle-levels}
A middle levels Gray code exists for every $n\geq 1$.
\end{theorem}

\subsection{Our results}

Even though the proof of Theorem~\ref{thm:middle-levels} given in \cite{MR3483129} is constructive, a straightforward implementation takes exponential (in $n$) time and space to compute for a given bitstring the next one in the middle levels Gray code.
Essentially, we need to compute and store the entire list of $N=2^{\Theta(n)}$ bitstrings.
The main contribution of this paper is a time- and space-efficient algorithm to compute a middle levels Gray code for every $n\geq 1$, which can be considered an algorithmic proof of Theorem~\ref{thm:middle-levels}.
Specifically, given any bitstring of length $2n+1$ with weight $n$ or $n+1$, our algorithm computes the next $\ell$ bitstrings in the Gray code in time $\cO(n\ell(1+\frac{n}{\ell}))$, which is $\cO(n)$ on average per bitstring provided that $\ell=\Omega(n)$.
For most bitstrings the worst-case running time is $\cO(n)$, for few it is $\cO(n^2)$. 
The space requirement of our algorithm is $\cO(n)$.

An implementation of this algorithm in C++ can be found on the authors' websites \cite{www}, and we invite the reader to experiment with this code. 
We used it to compute a middle levels Gray code for $n=19$ in less than a day on an ordinary desktop computer. 
For comparison, the above-mentioned intelligent brute-force computation for $n=19$ from \cite{shimada-amano} took about 164 days using comparable hardware.

\begin{remark}
Clearly, the ultimate goal would be a generation algorithm with a worst-case running time of $\cO(1)$ per bitstring, but this would require substantial new ideas that would in particular yield a much simpler proof of Theorem~\ref{thm:middle-levels} than the one presented in \cite{MR3483129}.
In fact, in a recent paper \cite{DBLP:conf/soda/MutzeN17} that appeared after the submission of this manuscript, we were able to extend the ideas and techniques presented in this paper, and to improve our algorithm to achieve an optimal $\cO(1)$ worst-case running time per generated bitstring.
This improved algorithm builds on top of the algorithm presented in this paper.
\end{remark}

\begin{remark}
It was shown in \cite{MR3483129} that there are in fact double-exponentially (in $n$) many different middle levels Gray codes, which is easily seen to be best possible.
This raises the question whether our algorithm can be parametrized to compute any of these Gray codes. 
While this is possible in principle, choosing between doubly-exponentially many different Gray codes would require a parameter of exponential size, spoiling the above-mentioned runtime and space bounds. 
Moreover, it would introduce a substantial amount of additional complexity in the description and correctness proof of the algorithm. 
To avoid all this, our algorithm computes only one particular `canonical' middle levels Gray code. 
We will briefly come back to this point at the end of Section~\ref{sec:correctness} (Remark~\ref{remark:other-cycles}).
\end{remark}

\subsection{Outline of this paper}

In Section~\ref{sec:algorithm} we present the pseudocode of our middle levels Gray code algorithm.
In Section~\ref{sec:correctness} we prove the correctness of this algorithm, and in Section~\ref{sec:running-time} we discuss how to implement it to achieve the claimed runtime and space bounds.

\section{The algorithm}
\label{sec:algorithm}

It is convenient to reformulate our problem in graph-theoretic language: To this end we define the \emph{middle levels graph}, denoted by $Q_{2n+1}(n,n+1)$, as the graph whose vertices are all bitstrings of length $2n+1$ that have weight $n$ or $n+1$, with an edge between any two bitstrings that differ in exactly one bit.
In other words, bitstrings that may appear consecutively in the middle levels Gray code correspond to neighboring vertices in the middle levels graph. 
Clearly, computing a middle levels Gray code is equivalent to computing a Hamilton cycle in the middle levels graph. 
Throughout the rest of this paper we talk about middle levels Gray codes using this graph-theoretic terminology.

Our algorithm to compute a Hamilton cycle in the middle levels graph, i.e., to compute a middle levels Gray code, is inspired by the constructive proof of Theorem~\ref{thm:middle-levels} given in \cite{MR3483129}. 
Efficiency is achieved by reformulating this inductive construction as a recursive procedure. 
Even though the description of the algorithm in the present paper is completely self-contained and illustrated with several figures that highlight the main ideas, the reader may find it useful to first read an informal overview of the proof of Theorem~\ref{thm:middle-levels}, which can be found in \cite[Section~1.2]{MR3483129}.

Roughly speaking, our algorithm consists of a lower level function that computes sets of disjoint paths in the middle levels graph and several higher level functions that combine these paths to form a Hamilton cycle. 
In the following we explain these functions from bottom to top. 
Before doing so we introduce some notation that will be used throughout this paper.

\subsection{Basic definitions}
\label{sec:basic-defs}

\textit{Composition of mappings.}
We write the composition of mappings $f,g$ as $f\bullet g$, where $(f\bullet g)(x):=f(g(x))$.

\textit{Reversing/inverting and concatenating bitstrings.}
We define $\ol{0}:=1$ and $\ol{1}:=0$.
For any bitstring $x$ we let $\ol{\rev}(x)$ denote the bitstring obtained from $x$ by reversing the order of the bits and inverting every bit. 
Moreover, for any bitstring $x=(x_1,x_2,\ldots,x_{2n})$ we define $\pi(x):=(x_1,x_3,x_2,x_5,x_4,\ldots,x_{2n-1},x_{2n-2},x_{2n})$, i.e., except the first and last bit, all adjacent pairs of bits are swapped.
E.g., for $x=11001110$ we have $\ol{\rev}(x)=10001100$ and $\pi(x)=10110110$.
Here and throughout this paper, we omit commas and brackets when writing strings of 0's and 1's.
Note that the mappings $\ol{\rev}$, $\pi$ and consequently also $\ol{\rev}\bullet \pi$ are self-inverse, and that $\ol{\rev}\bullet \pi=\pi\bullet\ol{\rev}$.
For two bitstrings $x$ and $y$ we denote by $x\circ y$ the concatenation of $x$ and $y$. 
For any graph $G$ whose vertices are bitstrings and any bitstring $y$ we denote by $G\circ y$ the graph obtained from $G$ by attaching $y$ to every vertex of $G$.

\textit{Layers of the cube.}
For $n\geq 0$ and $k$ with $0\leq k\leq n$, we denote by $B_n(k)$ the set of all bitstrings of length $n$ with weight $k$.
For $n\geq 1$ and $k$ with $0\leq k\leq n-1$, we denote by $Q_n(k,k+1)$ the graph with vertex set $B_n(k)\cup B_n(k+1)$, with an edge between any two bitstrings that differ in exactly one bit.

\textit{Oriented paths, first/second/last vertices.}
An oriented path $P$ in a graph is a path with a particular orientation, i.e., we distinguish its first and last vertex. 
For an oriented path $P=(v_1,v_2,\ldots,v_\ell)$ we define its first, second and last vertex as $F(P):=v_1$, $S(P):=v_2$ and $L(P):=v_\ell$, respectively.

\textit{Bitstrings and Dyck paths.}
We often identify a bitstring $x$ with a lattice path in the integer lattice $\mathbb{Z}^2$ as follows, see the left hand side of Figure~\ref{fig:bij}:
Starting at the coordinate $(0,0)$, we read the bits of $x$ from left to right and interpret every 1-bit as an upstep that changes the current coordinate by $(+1,+1)$ and every 0-bit as a downstep that changes the current coordinate by $(+1,-1)$.
Specifically, we define the following sets of lattice paths:
\begin{itemize}[topsep=0mm,leftmargin=4mm]
\item For any $n\geq 0$ and $k\geq 0$ we denote by $D_n(k)$ the set of lattice paths with $k$ upsteps and $n-k$ downsteps ($n$ steps in total) that never move below the line $y=0$.
\item For $n\geq 1$ and $k\geq 0$ we define $D_n^{>0}(k)\seq D_n(k)$ as the set of lattice paths that have no point of the form $(x,0)$, $1\leq x\leq n$, and $D_n^{=0}(k)\seq D_n(k)$ as the set of lattice paths that have at least one point of the form $(x,0)$, $1\leq x\leq n$.
For $n=0$ we define $D_0^{=0}(0):=\{()\}$, where $()$ denotes the empty lattice path, and $D_0^{>0}(0):=\emptyset$.
We clearly have $D_n(k)=D_n^{=0}(k)\cup D_n^{>0}(k)$.
\item Furthermore, for $n\geq 1$ and $k\geq 0$ we let $D_n^-(k)$ denote the set of lattice paths with $k$ upsteps and $n-k$ downsteps ($n$ steps in total) that have exactly one point of the form $(x,-1)$, $1\leq x\leq n$.
\end{itemize}
Note that all these lattice paths end at the coordinate $(n,2k-n)$.

For example, we have
\begin{align*}
  D_6^{=0}(3) &= \{111000,110100,110010,101010,101100\} \enspace, \\
  D_6^{>0}(4) &= \{111100,111010,111001,110110,110101\} \enspace,  \\
  D_6^-(3)    &= \{110001,101001,100110,011100,011010\} \enspace.
\end{align*}
It is well known that in fact $|D_{2n}^{=0}(n)|=|D_{2n}^{>0}(n+1)|=|D_{2n}^-(n)|$ and that the size of these sets is given by the $n$-th Catalan number (see \cite{Chen20081328} and \cite{MR1676282}).
Depending on the values of $n$ and $k$ the above-mentioned sets of lattice paths might be empty, e.g., we have $D_{2n}^{>0}(n)=\emptyset$. 
Observe furthermore that the mappings $\ol{\rev}$, $\pi$ and therefore also $\ol{\rev}\bullet \pi$ map each of the sets $D_{2n}^{=0}(n)$ and $D_{2n}^-(n)$ onto itself (see \cite[Lemma~11]{MR3483129} for a formal proof).
E.g., we have $\ol{\rev}(\pi(111000))=111000$, $\ol{\rev}(\pi(110100))=101010$, $\ol{\rev}(\pi(110010))=110010$, $\ol{\rev}(\pi(101010))=110100$, $\ol{\rev}(\pi(101100))=101100$ and therefore $\ol{\rev}(\pi(D_6^{=0}(3)))=D_6^{=0}(3)$.

\subsection{Computing paths in \texorpdfstring{$Q_{2n}(n,n+1)$}{Q2n(n,n+1)}}
\label{sec:paths}

The algorithm $\Paths()$ is at the core of our Hamilton cycle algorithm. 
Its description is given in Algorithm~\ref{alg:paths}.
For simplicity let us ignore for the moment the parameter $\flip\in\{\true,\false\}$ and assume that it is set to $\false$. 
Then for every $n$ and $k$ with $1\leq n \leq k \leq 2n-1$ the algorithm $\Paths()$ constructs a set of disjoint oriented paths $\cP_{2n}(k,k+1)$ in the graph $Q_{2n}(k,k+1)$ in the following way: 
Given a vertex $x\in Q_{2n}(k,k+1)$ and the parameter $\dir\in\{\prevv,\nextv\}$, the algorithm computes a neighbor of $x$ on the path that contains the vertex $x$. 
The parameter $\dir$ controls the search direction, so for $\dir=\prevv$ we obtain the neighbor of $x$ that is closer to the first vertex of the path, and for $\dir=\nextv$ the neighbor that is closer to the last vertex of the path. 
If $x$ is a first vertex or a last vertex of one of the paths --- we will see momentarily what these are --- then the result of a call to $\Paths()$ with $\dir=\prevv$ or $\dir=\nextv$, respectively, is undefined and such calls will not be made from the higher level functions. 

The algorithm $\Paths()$ works recursively: For the base case of the recursion $n=1$ (lines~\ref{line:n1}--\ref{line:base-case1}) it computes neighbors for the set of paths 
\begin{align}
\label{eq:P212}
  \cP_2(1,2):=\{(10,11,01)\} \enspace,
\end{align}
which consists only of a single path on three vertices. 
For example, the result of $\Paths(1,1,10,\nextv,\false)$ is $11$ and the result of $\Paths(1,1,11,\prevv,\false)$ is $10$.
For the recursion step the algorithm considers the last two bits of the current vertex $x$ (see line~\ref{line:split-x}) and, depending on their values (see lines~\ref{line:xp10}, \ref{line:xp00}, \ref{line:xp01} and \ref{line:xp11}), either flips one of these two bits (see lines~\ref{line:xp00f}, \ref{line:xp01f1}, \ref{line:xp01f2}, \ref{line:xp01f3}, \ref{line:xp11f1}, \ref{line:xp11f2}), or recurses to flip one of the first $2n-2$ bits instead, leaving the last two bits unchanged (see lines~\ref{line:recurse-upper-layers}, \ref{line:recurse-10}, \ref{line:recurse-00}, \ref{line:recurse-01}, \ref{line:recurse-11}).
As already mentioned, the algorithm $\Paths()$ is essentially a recursive formulation of the inductive construction of paths described in \cite{MR3483129} (and recapitulated in Section~\ref{sec:recap-paths}), and the different cases in the algorithm reflect the different cases in this construction.
The recursion step in line~\ref{line:recurse-11} is where the mappings $\ol{\rev}$ and $\pi$ introduced in Section~\ref{sec:basic-defs} come into play (recall that $\ol{\rev}^{-1}=\ol{\rev}$ and $\pi^{-1}=\pi$).
To give another example, the result of $\Paths(3,3,111000,\nextv,\false)$ is $111001$, and the result of $\Paths(3,3,111001,\prevv,\false)$ is $111000$, so the set of paths $\cP_6(3,4)$ in the graph $Q_6(3,4)$ defined by the algorithm $\Paths()$ contains a path with the edge $(111000,111001)$.

%\incmargin{1em}
\begin{algorithm}
\renewcommand\theAlgoLine{P\arabic{AlgoLine}}
\LinesNumbered
\DontPrintSemicolon
\SetEndCharOfAlgoLine{}
\SetNlSty{}{}{}
\SetArgSty{}
\SetKw{KwIf}{if}
\SetKw{KwElseIf}{else if}
\SetKw{KwElse}{else}
\SetKw{KwThen}{then}
\caption[Algorithm Paths()]{\mbox{$\Paths(n,k,x,\dir,\flip)$}}
\label{alg:paths}
\vspace{.2em}
\KwIn{Integers $n$ and $k$ with $1\leq n\leq k\leq 2n-1$, a vertex $x\in Q_{2n}(k,k+1)$, parameters $\dir\in\{\prevv,\nextv\}$ and $\flip\in\{\true,\false\}$} 
\KwOut{A neighbor of $x$ in $Q_{2n}(k,k+1)$} 
\vspace{.2em}
\If (\tcc*[f]{base cases}) {$n=1$ \label{line:n1}} { 
    depending on $x$ and $\dir$, \Return the previous/next neighbor of $x$ on $\cP_{2}(1,2)$ defined in \eqref{eq:P212} \label{line:base-case1}
}
\ElseIf {$n=k=2$ and $\flip=\true$  \label{line:n2}} {
    depending on $x$ and $\dir$, \Return the previous/next neighbor of $x$ on $\tcP_4(2,3)$ defined in \eqref{eq:tP423} \label{line:base-case2}
}
Split $x=(x_1,x_2,\ldots,x_{2n})$ into $x^-:=(x_1,x_2,\ldots,x_{2n-2})$ and $x^+:=(x_{2n-1},x_{2n})$ \label{line:split-x} \;
\If {$k \geq n+1$ \label{line:upper-layers-cond}}  {  
  \Return $\Paths(n-1,k-x_{2n-1}-x_{2n},x^-,\dir,\flip)\circ x^+$ \label{line:recurse-upper-layers} \;
}
\Else (\tcc*[f]{$k=n$}) {  \label{line:recurse-main-layer}
  \If {$x^+=10$ \label{line:xp10}} {
    \Return $\Paths(n-1,n-1,x^-,\dir,\flip)\circ x^+$ \label{line:recurse-10} \;
  } 
  \ElseIf {$x^+=00$ \label{line:xp00}} {
    \KwIf $x^- \in D_{2n-2}^{>0}(n)$ \KwThen \Return $x^-\circ 01$ \label{line:xp00f} \;
    \KwElse \Return $\Paths(n-1,n,x^-,\dir,\flip)\circ x^+$ \label{line:recurse-00}
  }
  \ElseIf {$x^+=01$ \label{line:xp01}} {
    \KwIf $x^- \in D_{2n-2}^{=0}(n-1)$ \KwThen \Return $x^-\circ 11$ \label{line:xp01f1} \;
    \KwElseIf $x^- \in D_{2n-2}^{-}(n-1)$ and $\dir=\nextv$ \KwThen \Return $x^-\circ 11$ \label{line:xp01f2} \;
    \KwElseIf $x^- \in D_{2n-2}^{>0}(n)$ and $\dir=\prevv$ \KwThen \Return $x^-\circ 00$ \label{line:xp01f3} \;
    \KwElse \Return $\Paths(n-1,n-1,x^-,\dir,\false)\circ x^+$ \label{line:recurse-01}
  }
  \ElseIf {$x^+=11$ \label{line:xp11}} {
    \KwIf $x^- \in D_{2n-2}^{=0}(n-1)$ and $\dir=\nextv$ \Return $x^-\circ 01$ \label{line:xp11f1} \;
    \KwElseIf $x^- \in D_{2n-2}^{-}(n-1)$ and $\dir=\prevv$ \KwThen \Return $x^-\circ 01$ \label{line:xp11f2} \;
    \KwElse \Return $\ol{\rev}\left(\pi\left(\Paths(n-1,n-1,\ol{\rev}^{-1}(\pi^{-1}(x^-)),\ol{\dir},\flip)\right)\right)\circ x^+$ where  $\ol{\dir}:=\prevv$ if $\dir=\nextv$ and $\ol{\dir}:=\nextv$ otherwise \label{line:recurse-11}
  }
} 
\end{algorithm}
%\decmargin{1em}

We will later prove that the set of paths $\cP_{2n}(n,n+1)$ computed by the algorithm $\Paths()$ (with parameters $k=n$ and $\flip=\false$) has the following properties:
\begin{enumerate}[label=(\roman*),topsep=0mm,leftmargin=5mm]
\item All paths in $\cP_{2n}(n,n+1)$ are disjoint, and together they visit all vertices of the graph $Q_{2n}(n,n+1)$.
\item The sets of first and last vertices of the paths in $\cP_{2n}(n,n+1)$ are $D_{2n}^{=0}(n)$ and $D_{2n}^{-}(n)$, respectively.
\end{enumerate}
We only state these properties here; the proof will be provided in Section~\ref{sec:correctness}.
From (ii) we conclude that the number of paths in $\cP_{2n}(n,n+1)$ equals the $n$-th Catalan number.
E.g., the entire set of paths $\cP_4(2,3)$ computed by the algorithm $\Paths()$ is
\begin{subequations}
\label{eq:P423}
\begin{align}
  \cP_4(2,3) &= \{P,P'\} \enspace, \\
  P  &:= (1100,1101,0101,0111,0011,1011,1001) \enspace, \\
  P' &:= (1010,1110,0110) \enspace.
\end{align}
\end{subequations}

\begin{table}
\begin{tabular}{c*{5}{c}}
$P_1$ & $P_2$ & $P_3$ & $P_4$ & $P_5$ \\ \hline
111000	&	110100	&	110010	&	101010	&	101100	\\
111001	&	110101	&	110110	&	111010	&	111100	\\
011001	&	010101	&	010110	&	011010	&	011100	\\
011011	&	011101	&	011110	& & \\
010011	&	001101	&	001110	& & \\
010111	&	101101	&	101110	& & \\
000111	&	100101	&	100110	& &	\\
001111	&	100111	& & &	\\
001011	&	100011	& & &	\\
101011	&	110011	& & &	\\
101001	&	110001	& & &		
\end{tabular}
\caption{The five paths in the set $\cP_6(3,4)$ computed by the algorithm $\Paths()$ with parameter $\flip=\false$.
Each path is listed in a separate column, the first vertex from the set $D_6^{=0}(3)$ is at the top, and the last vertex from the set $D_6^-(3)$ is at the bottom.
These paths are disjoint and together visit all vertices of the graph $Q_6(3,4)$.
}
\label{tab:P634}
\end{table}

\begin{table}
\begin{tabular}{c*{4}{c}}
$R_1$ & $R_2$ & $R_3$ & $R_4$ \\ \hline
111000	&	110100	&	110010		&	101010	\\
111001	&	110101	&	111010		&	101110	\\
011001	&	010101	&	011010		&	001110	\\
011011	&	011101	&		&			011110	\\
001011	&	001101	&		&			010110	\\
001111	&	101101	&		&			110110	\\
000111	&	100101	&		&			100110	\\
010111	&	100111	&		&				\\
010011	&	100011	&		&				\\
110011	&	101011	&		&				\\
110001	&	101001	&		&				
\end{tabular}
\caption{The four paths in the set $\tcP_6(3,4)$ computed by the algorithm $\Paths()$ with parameter $\flip=\true$.
Comparing those paths to Table~\ref{tab:P634}, note that $(P_1,P_2)$ and $(P_3,P_4)$ are flippable pairs of paths and the corresponding flipped pairs of paths are $(R_1,R_2)$ and $(R_3,R_4)$, respectively.
Note that the paths in $\tcP_6(3,4)$ do not visit all vertices of $Q_6(3,4)$ (unlike the paths in $\cP_6(3,4)$), as the vertices of $P_5$ from Table~\ref{tab:P634} are not covered.
}
\label{tab:tP634}
\end{table}

To give another example, the set of paths $\cP_6(3,4)$ computed by the algorithm $\Paths()$ consists of the five paths in Table~\ref{tab:P634}.
We remark that the length of each of the paths in $\cP_{2n}(n,n+1)$ depends only on the length of the first `hill' of the lattice path corresponding to the first vertex.
In particular, not all paths have the same length.
The precise length formula will be given in Section~\ref{sec:running-time}, when we analyze the running time of our algorithm.
For the correctness of the algorithm the lengths of the paths are irrelevant.
We will see in the next section how to combine the paths $\cP_{2n}(n,n+1)$ in the graph $Q_{2n}(n,n+1)$ computed by our algorithm $\Paths()$ (called with $\flip=\false$) to compute a Hamilton cycle in the middle levels graph $Q_{2n+1}(n,n+1)$. 
One crucial ingredient we need for this is another set of paths $\tcP_{2n}(n,n+1)$, computed by calling the algorithm $\Paths()$ with $\flip=\true$. 
We shall see that the core `intelligence' of our Hamilton cycle algorithm consists of cleverly combining some paths from $\cP_{2n}(n,n+1)$ and some paths from $\tcP_{2n}(n,n+1)$ to a Hamilton cycle in the middle levels graph. 
We will see that from the point of view of all top-level routines that call the algorithm $\Paths()$ only the paths $\cP_{2n}(k,k+1)$ and $\tcP_{2n}(k,k+1)$ for $k=n$ are used, but the recursion clearly needs to compute these sets also for all other values $k=n+1,n+2,\ldots,2n-1$.

Specifically, calling the algorithm $\Paths()$ with $\flip=\true$ yields a set of paths $\tcP_{2n}(k,k+1)$ that differs from $\cP_{2n}(k,k+1)$ as follows:
First note that for certain inputs of the algorithm the value of the parameter $\flip$ is irrelevant, and the computed paths are the same regardless of its value.
This is true whenever the recursion ends in the base case in lines~\ref{line:n1}-\ref{line:base-case1}. 
One such example is the path $P_5$ in Table~\ref{tab:P634}.
Computing previous or next neighbors for any vertex on this path yields the same result regardless of the value of $\flip$.
Such paths are not part of $\tcP_{2n}(k,k+1)$, and we ignore them in the following.
The remaining paths from $\cP_{2n}(k,k+1)$ can be grouped into pairs and for every such pair $(P,P')$ the set $\tcP_{2n}(k,k+1)$ contains two disjoint paths $R$ and $R'$ that visit the same set of vertices as $P$ and $P'$, but that connect the end vertices of the paths the other way: Formally, denoting by $V(G)$ the vertex set of any graph $G$, these last two conditions can be captured as $V(P)\cup V(P')=V(R)\cup V(R')$, $F(P)=F(R)$, $F(P')=F(R')$, $L(P)=L(R')$ and $L(P')=L(R)$.

We refer to a pair $(P,P')$ satisfying these conditions as a \emph{flippable pair} of paths, and to the corresponding paths $(R,R')$ as a \emph{flipped pair}. 
This notion of flippable/flipped pairs is extremely valuable for designing our Hamilton cycle algorithm, as it allows the algorithm to decide \emph{independently} for each flippable pair of paths, whether to follow one of the original paths or the flipped paths in the graph $Q_{2n}(n,n+1)$.

The base case for the recursive computation of $\tcP_{2n}(k,k+1)$ is the set of paths $\tcP_4(2,3)$ in the graph $Q_4(2,3)$ defined by
\begin{subequations}
\label{eq:tP423}
\begin{align}
  \tcP_4(2,3) &:= \{R,R'\} \enspace, \\
  R  &:= (1100,1110,0110) \enspace, \\
  R' &:= (1010,1011,0011,0111,0101,1101,1001)
\end{align}
\end{subequations}
(see lines~\ref{line:n2}--\ref{line:base-case2}).
Observe that $(P,P')$ defined in \eqref{eq:P423} and $(R,R')$ satisfy precisely the conditions for flippable/flipped pairs of paths, i.e., $(P,P')$ is a flippable pair and $(R,R')$ is a corresponding flipped pair of paths.
To continue also the previous example, the set of paths $\tcP_6(3,4)$ computed by our algorithm is shown in Table~\ref{tab:tP634}.

\subsection{Computing a Hamilton cycle in the middle levels graph}
\label{sec:hamcyclenext}

The algorithm $\HamCycleNext()$ uses the paths $\cP_{2n}(n,n+1)$ and $\tcP_{2n}(n,n+1)$ in the graph $Q_{2n}(n,n+1)$ computed by the algorithm $\Paths()$ as described in the previous section to compute for a given vertex $x$ the next vertex on a Hamilton cycle in the middle levels graph $Q_{2n+1}(n,n+1)$. 
To simplify the exposition of the algorithm, let us ignore for the moment the parameter $\flip\in\{\true,\false\}$ and assume that it is set to $\false$, and let us also ignore line~\ref{line:check-flip}.
With these simplifications only paths from $\cP_{2n}(n,n+1)$ are considered, and those from $\tcP_{2n}(n,n+1)$ are ignored. 
Then instead of computing a Hamilton cycle the algorithm $\HamCycleNext()$ computes several smaller cycles that together visit all vertices of the middle levels graph.
We will correct this later by setting $\flip$ accordingly.

The algorithm $\HamCycleNext()$ is based on the following decomposition of the middle levels graph $Q_{2n+1}(n,n+1)$, see Figure~\ref{fig:2factor}:
By partitioning the vertices of the graph $Q_{2n+1}(n,n+1)$ according to the value of the last bit, we observe that it consists of a copy of the graph $Q_{2n}(n,n+1)\circ 0$ and a copy of the graph $Q_{2n}(n-1,n)\circ 1$ plus the set of edges $M_{2n+1}=\{(x\circ 0 ,x\circ 1 )\mid x\in B_{2n}(n)\}$ along which the last bit is flipped.
Observe furthermore that the graphs $Q_{2n}(n,n+1)$ and $Q_{2n}(n-1,n)$ are isomorphic, and that the mapping $\ol{\rev}$ is an isomorphism between these graphs. 
It is easy to check that this isomorphism preserves the sets of end vertices of the paths $\cP_{2n}(n,n+1)$:
Using property~(ii) mentioned in the previous section, we have $\ol{\rev}(D_{2n}^{=0}(n))=D_{2n}^{=0}(n)$ and $\ol{\rev}(D_{2n}^-(n))=D_{2n}^-(n)$ (in Figure~\ref{fig:2factor}, these sets are the black and white vertices).
By property~(i) of the paths $\cP_{2n}(n,n+1)$, we conclude that
\begin{equation}
\label{eq:2-factor}
  \cC_{2n+1}:=\cP_{2n}(n,n+1)\circ 0 \;\cup\; \ol{\rev}(\cP_{2n}(n,n+1))\circ 1 \;\cup\; M_{2n+1}'
\end{equation}
with $M_{2n+1}':=\{(x\circ 0,x\circ 1)\mid x\in D_{2n}^{=0}(n)\cup D_{2n}^-(n)\}\seq M_{2n+1}$ is a so-called \emph{2-factor} of the middle levels graph, i.e., a set of disjoint cycles that together visit all vertices of the graph. 
Note that along each of the cycles in the 2-factor, the paths from $\cP_{2n}(n,n+1)\circ 0$ are traversed in forward direction, and the paths from $\ol{\rev}(\cP_{2n}(n,n+1))\circ 1$ in backward direction. 
The algorithm $\HamCycleNext()$ (called with $\flip=\false$) computes exactly this 2-factor $\cC_{2n+1}$: Given a vertex $x$ of the middle levels graph, it computes the next vertex on one of the cycles from the 2-factor by checking the value of the last bit (line~\ref{line:last-bit0}), and by returning either the next vertex on the corresponding path from $\cP_{2n}(n,n+1)\circ 0$ (line~\ref{line:paths-next}) or the previous vertex on the corresponding path from $\ol{\rev}(\cP_{2n}(n,n+1))\circ 1$ (line~\ref{line:paths-prev}; recall that $\ol{\rev}^{-1}=\ol{\rev}$). 
The cases that the next cycle edge is an edge from $M_{2n+1}'$ receive special treatment: In these cases the last bit is flipped (lines~\ref{line:flip-last-bit1} and \ref{line:flip-last-bit0}).

\begin{figure}
\centering
\PSforPDF{
 \psfrag{cc}{\Large $\cC_{2n+1}$}
 \psfrag{q2n0}{$Q_{2n}(n,n+1)\circ 0$}
 \psfrag{q2n1}{$Q_{2n}(n-1,n)\circ 1$}
 \psfrag{q2np1}{$Q_{2n+1}(n,n+1)$}
 \psfrag{b2nnp1}{$B_{2n}(n+1)\circ 0$}
 \psfrag{b2n0}{$B_{2n}(n)\circ 0$}
 \psfrag{b2n1}{$B_{2n}(n)\circ 1$}
 \psfrag{b2nnm1}{$B_{2n}(n-1)\circ 1$}
 \psfrag{m2}{$M_{2n+1}$}
 \psfrag{pa}{$\cP_{2n}(n,n+1)$}
 \psfrag{pap}{$\ol{\rev}(\cP_{2n}(n,n+1))$}
 \includegraphics{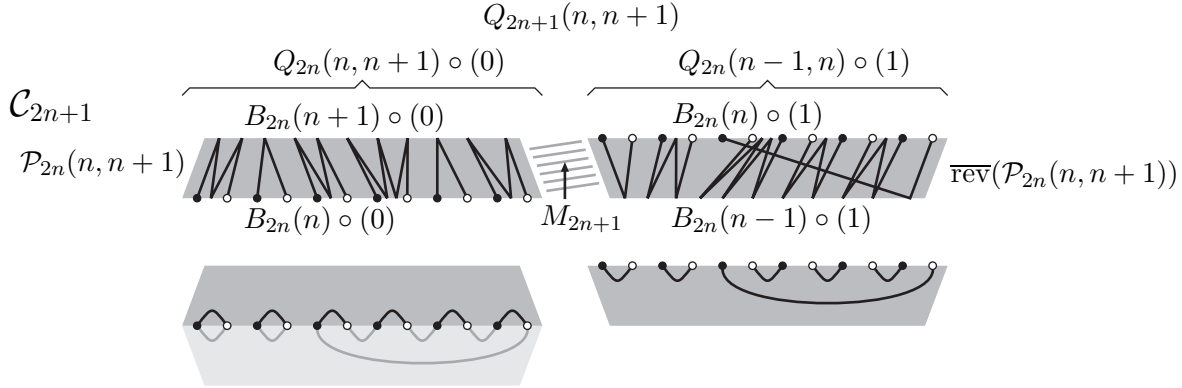}
}
\caption{The top part of the figure shows the decomposition of the middle levels graph and the definition~\eqref{eq:2-factor}. 
The 2-factor consists of three disjoint cycles that together visit all vertices of the graph. 
The first and last vertices of the paths are drawn in black and white, respectively. 
The bottom part of the figure shows a simplified drawing that helps analyzing the cycle structure of the 2-factor (it has two short cycles and one long cycle).}
\label{fig:2factor}
\end{figure}

\begin{algorithm}
\renewcommand\theAlgoLine{N\arabic{AlgoLine}}
\LinesNumbered
\DontPrintSemicolon
\SetEndCharOfAlgoLine{}
\SetNlSty{}{}{}
\SetArgSty{}
\SetKw{KwIf}{if}
\SetKw{KwElseIf}{else if}
\SetKw{KwElse}{else}
\SetKw{KwThen}{then}
\SetKw{KwWhile}{while}
\SetKw{KwDo}{do}
\caption[Algorithm HamCycleNext()]{\mbox{$\HamCycleNext(n,x,\flip)$}}
\label{alg:hamcyclenext}
\vspace{.2em}
\KwIn{An integer $n\geq 1$, a vertex $x\in Q_{2n+1}(n,n+1)$, state variable $\flip\in\{\true,\false\}$} 
\KwOut{Starting from $x$ the next vertex on a Hamilton cycle in $Q_{2n+1}(n,n+1)$, updated state variable $\flip\in\{\true,\false\}$}
\vspace{.2em}
Split $x=(x_1,x_2,\ldots,x_{2n+1})$ into $x^-:=(x_1,x_2,\ldots,x_{2n})$ and the last bit $x_{2n+1}$ \;
\If {$x_{2n+1}=0$ \label{line:last-bit0} } {
  \KwIf $x^- \in D_{2n}^{=0}(n)$ \KwThen
    \Return $\big(\Paths(n,n,x^-,\nextv,a)\circ x_{2n+1},a\big)$ where $a:=\IsFlipVertex(n,x^-)$ \label{line:check-flip} \;
  \KwElseIf $x^- \in D_{2n}^{-}(n)$ \KwThen
    \Return $(x^-\circ 1,\false)$ \label{line:flip-last-bit1} \;
  \KwElse
    \Return $\big(\Paths(n,n,x^-,\nextv,\flip)\circ x_{2n+1},\flip\big)$ \label{line:paths-next} \;
}
\Else (\tcc*[f]{$x_{2n+1}=1$} \label{line:last-bit1} ) {
  \KwIf $x^- \in D_{2n}^{=0}(n)$ \KwThen \Return $(x^-\circ 0,\false)$ \label{line:flip-last-bit0} \;
  \KwElse \Return $\big(\ol{\rev}\left(\Paths(n,n,\ol{\rev}^{-1}(x^-),\prevv,\false)\right)\circ x_{2n+1},\false\big)$ \label{line:paths-prev}
}
\end{algorithm}
%\decmargin{1em}

As mentioned before, calling the algorithm $\HamCycleNext()$ with $\flip=\false$ yields the 2-factor in the middle levels graph defined in \eqref{eq:2-factor} that consists of more than one cycle.
However, we will show that by replacing in \eqref{eq:2-factor} some of the flippable pairs of paths from $\cP_{2n}(n,n+1)\circ 0$ by the corresponding flipped paths from the set $\tcP_{2n}(n,n+1)\circ 0$, which are computed by calling the algorithm $\Paths()$ with $\flip=\true$, we obtain a 2-factor that has only one cycle, i.e., a Hamilton cycle. 
The key insight is that replacing a pair of flippable paths that are contained in two different cycles of the 2-factor $\cC_{2n+1}$ by the corresponding flipped paths joins these two cycles to one cycle.
This is an immediate consequence of the definition of flippable/flipped pairs of paths.
This flipping is controlled by the parameter $\flip$ of the algorithm $\HamCycleNext()$: It decides whether to compute the next vertex on a path from $\cP_{2n}(n,n+1)\circ 0$ (if $\flip=\false$) or from $\tcP_{2n}(n,n+1)\circ 0$ (if $\flip=\true$). 
Note that these modifications do not affect the paths $\ol{\rev}(\cP_{2n}(n,n+1))\circ 1$ in the union \eqref{eq:2-factor}: For the corresponding instructions in lines~\ref{line:last-bit1}--\ref{line:paths-prev}, the value of $\flip$ is irrelevant.

The decision whether to follow a path from $\cP_{2n}(n,n+1)\circ 0$ or from $\tcP_{2n}(n,n+1)\circ 0$ is computed at the vertices $x\circ 0$, $x\in D_{2n}^{=0}(n)$, by calling the function $\IsFlipVertex(n,x)$ (line~\ref{line:check-flip}). 
This decision is returned to the caller and maintained until the last vertex of the corresponding path in the graph $Q_{2n}(n,n+1)\circ 0$ is reached.
Recall that by the definition of flippable pairs of paths, this decision can be made \emph{independently} for each flippable pair.
Of course it has to be consistent for both paths in a flippable pair: either both are flipped or none of them are.

\subsection{The top-level algorithm}
\label{sec:hamcycle}

The algorithm $\HamCycle(n,x,\ell)$ takes as input a vertex $x$ of the middle levels graph $Q_{2n+1}(n,n+1)$ and computes the next $\ell$ vertices that follow $x$ on a Hamilton cycle in this graph.
I.e., for $\ell\leq N=2\binom{2n+1}{n}$, every vertex appears at most once in the output, and for $\ell=N$ every vertex appears exactly once and the vertex $x$ comes last.
In terms of Gray codes, the algorithm takes a bitstring $x$ of length $2n+1$ that has weight $n$ or $n+1$ and outputs the $\ell$ subsequent bitstrings in a middle levels Gray code. 
A single call $\HamCycle(n,x,\ell)$ yields the same output as $\ell$ subsequent calls $x_{i+1}:=\HamCycle(n,x_i,1)$, $i=0,1,\ldots,\ell-1$, with $x_0:=x$. 
However, with respect to running times, the former is faster than the latter.

\begin{algorithm}
\renewcommand\theAlgoLine{H\arabic{AlgoLine}}
\LinesNumbered
\DontPrintSemicolon
\SetEndCharOfAlgoLine{}
\SetNlSty{}{}{}
\SetArgSty{}
\SetKw{KwIf}{if}
\SetKw{KwElseIf}{else if}
\SetKw{KwElse}{else}
\SetKw{KwThen}{then}
\SetKw{KwWhile}{while}
\SetKw{KwDo}{do}
\SetKw{KwOutput}{output}
\caption[Algorithm HamCycle()]{\mbox{$\HamCycle(n,x,\ell)$}}
\label{alg:hamcycle}
\vspace{.2em}
\KwIn{An integer $n\geq 1$, a vertex $x\in Q_{2n+1}(n,n+1)$, an integer $\ell\geq 1$}
\KwOut{Starting from $x$, the next $\ell$ vertices on a Hamilton cycle in $Q_{2n+1}(n,n+1)$}
\vspace{.2em}
Split $x=(x_1,x_2,\ldots,x_{2n+1})$ into $x^-:=(x_1,x_2,\ldots,x_{2n})$ and the last bit $x_{2n+1}$ \;
$\flip:=\false$ \label{line:init-flip-start} \;
\If (\tcc*[f]{initialize state variable $\flip$}) {$x_{2n+1}=0$} {
  $y:=x^-$ \;
  \While (\tcc*[f]{move backwards to first path vertex}) {$y\notin D_{2n}^{=0}(n)$ \label{line:while-not-start} } {
    $y:=\Paths(n,n,y,\prevv,\false)$ \label{line:go-back}
  }
  $\flip:=\IsFlipVertex(n,y)$ \label{line:init-flip-end}
}
$y:=x$ \label{line:hamcycle-loop-start} \;
\For (\tcc*[f]{Hamilton cycle computation}) {$i:=1$ \KwTo $\ell$ \label{line:hamcycle-for-loop}} {
  $(y,\flip):=\HamCycleNext(n,y,\flip)$ \label{line:call-hamcyclenext} \;
  \KwOutput $y$ \label{line:hamcycle-loop-end}
}
\end{algorithm}
%\decmargin{1em}

The algorithm $\HamCycle()$ consists of an initialization phase (lines~\ref{line:init-flip-start}--\ref{line:init-flip-end}) in which the initial value of the state variable $\flip$ is computed. 
This is achieved by following the corresponding path from $\cP_{2n}(n,n+1)\circ 0$ backwards to its first vertex (lines~\ref{line:while-not-start}--\ref{line:go-back}) and by calling the function $\IsFlipVertex()$ (line~\ref{line:init-flip-end}). 
The actual Hamilton cycle computation (lines~\ref{line:hamcycle-loop-start}--\ref{line:hamcycle-loop-end}) repeatedly computes the subsequent cycle vertex and updates the state variable $\flip$ by calling the function $\HamCycleNext()$ discussed in the previous section (line~\ref{line:call-hamcyclenext}).

\subsection{Flip vertex computation}
\label{sec:isflipvertex}

To complete the description of our Hamilton cycle algorithm, it remains to specify the auxiliary function $\IsFlipVertex(n,x)$. 
As mentioned before, this function decides for each vertex $x\in D_{2n}^{=0}(n)$ whether our Hamilton cycle algorithm should follow the path from $\cP_{2n}(n,n+1)$ that starts with this vertex (return value $\false$) or the corresponding flipped path from $\tcP_{2n}(n,n+1)$ (return value $\true$) in the graph $Q_{2n}(n,n+1)\circ 0$ (recall \eqref{eq:2-factor}).
This function therefore nicely encapsulates the core `intelligence' of our algorithm so that it produces a 2-factor consisting only of a single cycle and not of several smaller cycles.
Admittedly, the definition of this function is rather technical, and we postpone the formal correctness proof until Section~\ref{sec:correctness}. 
We begin by introducing several new concepts related to trees.

\textit{Ordered rooted trees.}
An \emph{ordered rooted tree} is a rooted tree where the children of each vertex have a specified left-to-right ordering. 
We think of an ordered rooted tree as a tree embedded in the plane with the root on top, with downward edges leading from any vertex to its children, and the children appear in the specified left-to-right ordering.
This is illustrated on the right hand side of Figure~\ref{fig:bij}, where the root vertex is drawn boldly.
We denote by $\cT_n^*$ the set of all ordered rooted trees with $n$ edges. 
It is well known that the number of trees in $\cT_n^*$ is given by the $n$-th Catalan number (see \cite{MR1676282}).

\textit{Bijection between lattice paths and ordered rooted trees.}
We identify each lattice path (=bitstring) $x=(x_1,x_2,\ldots,x_{2n})$ from the set $D_{2n}^{=0}(n)$ with an ordered rooted tree from the set $\cT_n^*$ as follows (see the right hand side of Figure~\ref{fig:bij} and \cite{MR1676282}):
Starting with a tree that has only a root vertex, we read the bits of $x$ from left to right, and for every 1-bit we add a new rightmost child to the current vertex and move to this child, for every 0-bit we add no edge but simply move back to the parent of the current vertex (we clearly end up back at the root).
This construction defines a bijection between the lattice paths $D_{2n}^{=0}(n)$ and the ordered rooted trees $\cT_n^*$, and in the following we will repeatedly switch between these two representations.

\begin{figure}
\centering
\PSforPDF{
 \psfrag{x}{\parbox{5cm}{bitstring \\ $x=1101101000$}}
 \psfrag{p}{\parbox{3cm}{lattice path \\ from $D_{10}^{=0}(5)$}}
 \psfrag{t}{\parbox{4cm}{ordered rooted \\ tree from $\cT_5^*$}}
 \psfrag{z}{0}
 \psfrag{ten}{10}
 \includegraphics{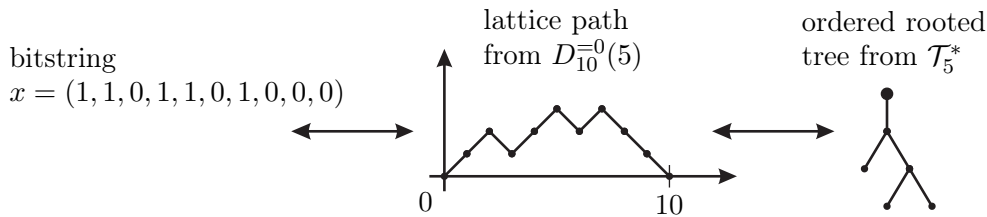}
}
\caption{Bijections between bitstrings and lattice paths (left hand side), and between lattice paths from $D_{2n}^{=0}(n)$ and ordered rooted trees from $\cT_n^*$ (right hand side).}
\label{fig:bij}
\end{figure}

\textit{Thin/thick leaves, clockwise/counterclockwise-next leaves.}
We call a leaf $u$ of a tree \emph{thin} or \emph{thick}, if the vertex adjacent to $u$ has degree exactly 2 or at least 3, respectively. 
Clearly, for any tree with at least two edges, every leaf is either thin or thick.
Given two leaves $u$ and $v$ of an ordered rooted tree $T$, we say that $v$ is the \emph{clockwise-next} leaf from $u$, or equivalently, $u$ is the \emph{counterclockwise-next} leaf from $v$, if all edges of $T$ not on the path $p$ from $u$ to $v$ lie to the right of the path $p$ when traversing $p$ from $u$ to $v$ in the given embedding of $T$, respectively.
Equivalently, in a (cyclic) depth-first search starting at the root, $u$ is the first leaf encountered after the leaf $v$.
This definition is illustrated in Figure~\ref{fig:tau12}.

\textit{Tree transformations $\tau_1$ and $\tau_2$.}
We define sets of ordered rooted trees $\cT_{n,1}^*,\cT_{n,2}^*\seq\cT_n^*$ and transformations $\tau_1,\tau_2$ operating on these trees as follows:
The set $\cT_{n,1}^*$ contains all ordered rooted trees $T$ with at least three edges for which the leftmost child $u'$ of the root $v$ has exactly one child $u$ and $u$ is a leaf.
In other words, the leftmost subtree of the root is a path on two edges, see the left hand side of Figure~\ref{fig:tau12}.
For such a tree $T$, we define $\tau_1(T)$ as the tree obtained by replacing the edge $(u,u')$ by $(u,v)$ so that $u$ becomes the leftmost child of the root.
In the bitstring representation, we have $T=1100\circ T'$ for some $T'\in\cT_{n-2}^*$, and $\tau_1(T)=1010\circ T'$.
The set $\cT_{n,2}^*$ contains all ordered rooted trees $T$ for which the rightmost child $u$ of the leftmost child $u'$ of the root is a thick leaf (so $u'$ has more children to the left of $u$), and the clockwise-next leaf $v$ of $u$ is thin.
Clearly, such a tree must have at least four edges, see the right hand side of Figure~\ref{fig:tau12}.
For such a tree $T$, we define $\tau_2(T)$ as the tree obtained by replacing the edge $(u,u')$ by $(u,v)$.
In the bitstring representation, we have
\begin{equation*}
  T=1\circ T_1\circ 1\circ T_2\circ 1\circ \cdots \circ T_{k+1}\circ 1100\circ 0^k\circ 10\circ 0\circ T_0
\end{equation*}
for some $k\geq 0$ and $T_i\in \cT_{n_i}^*$ for all $i=0,1,\ldots,k+1$, and
\begin{equation*}
  \tau_2(T)=1\circ T_1\circ 1\circ T_2\circ 1\circ \cdots \circ T_{k+1}\circ 111000\circ 0^k\circ 0\circ T_0 \enspace.
\end{equation*}
Note that the four sets $\cT_{n,1}^*$, $\cT_{n,2}^*$, $\tau_1(\cT_{n,1}^*)$ and $\tau_2(\cT_{n,2}^*)$ are all disjoint.

\begin{figure}
\centering
\PSforPDF{
 \psfrag{t}{$T$}
 \psfrag{tp}{$T'$}
 \psfrag{tau1}{$\tau_1(T)$}
 \psfrag{tau2}{$\tau_2(T)$}
 \psfrag{u}{$u$}
 \psfrag{up}{$u'$}
 \psfrag{v}{$v$}
 \psfrag{deg2}{$\deg=2$}
 \psfrag{deg3}{$\deg\geq 3$}
 \psfrag{cw}{$v$ is clockwise-next leaf from $u$}
 \psfrag{ccw}{$u$ is counterclockwise-next leaf from $v$}
 \includegraphics{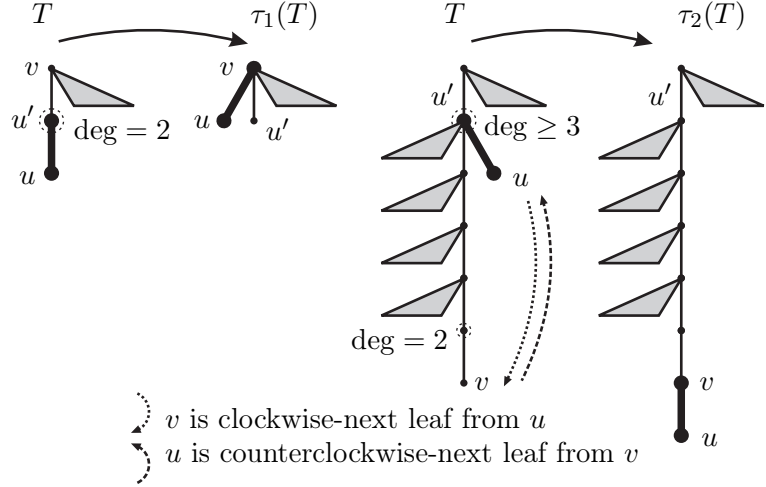}
}
\caption{Definition of the tree transformations $\tau_1$ (left) and $\tau_2$ (right). 
The edges in which the trees differ are drawn bold, and the grey areas represent arbitrary subtrees.}
\label{fig:tau12}
\end{figure}

\textit{The mapping $h$.}
For any lattice path $x\in D_{2n}^{=0}(n)$ we inductively define a lattice path $h(x)\in D_{2n}^{=0}(n)$ as follows:
If $x=()$ then we define
\begin{subequations}
\label{eq:def-h}
\begin{equation}
\label{eq:def-h-indbase}
  h(x):=() \enspace.
\end{equation}
Otherwise there is a unique partition $x=1\circ x_\ell\circ 0 \circ x_r$ such that $x_\ell\in D_{2k}^{=0}(k)$ for some $k\geq 0$ and we define
\begin{equation}
\label{eq:def-h-indstep}
  h(x):= 1\circ \pi(h(x_\ell))\circ 0\circ h(x_r) \enspace,
\end{equation}
where $\pi$ is the permutation defined in Section~\ref{sec:basic-defs}.
\end{subequations}
It is easy to check that $h$ defines a bijection on $D_{2n}^{=0}(n)$, and that the definition of the inverse mapping $h^{-1}$ can be derived from \eqref{eq:def-h} simply by replacing $\pi(h(x_\ell))$ in \eqref{eq:def-h-indstep} by $h^{-1}(\pi(x_\ell))$ and $h(x_r)$ by $h^{-1}(x_r)$.
Clearly, via the bijection between lattice paths and ordered rooted trees discussed before, the mapping $h$ can be interpreted as a bijection on the set $\cT_n^*$. 
As an example, Figure~\ref{fig:rot-trees} shows the effect of applying $h$ to all 14 trees in the set $\cT_4^*$.
The preimages under $h$ appear at the top of the figure, the corresponding images below.

With these definitions at hand, consider now the definition of the function $\IsFlipVertex(n,x)$.
It first computes the ordered rooted tree $T$ obtained by applying $h^{-1}$ to the bitstring/lattice path $x$ (line~\ref{line:h-inverse}).
On lines~\ref{line:tau_1-preimage}--\ref{line:tau_1-image} the function checks whether $T$ is a preimage or image under $\tau_1$ and calls an auxiliary function $\IsFlipTree_1()$ (defined below) with $T$ or $\tau_1^{-1}(T)$ as an argument, respectively.
A similar check is performed in lines~\ref{line:tau_2-preimage}--\ref{line:tau_2-image} with respect to $\tau_2$, using the auxiliary function $\IsFlipTree_2()$ (defined below).

\begin{algorithm}
\renewcommand\theAlgoLine{T\arabic{AlgoLine}}
\LinesNumbered
\DontPrintSemicolon
\SetEndCharOfAlgoLine{}
\SetNlSty{}{}{}
\SetArgSty{}
\SetKw{KwIf}{if}
\SetKw{KwElseIf}{else if}
\SetKw{KwElse}{else}
\SetKw{KwThen}{then}
\caption[Algorithm IsFlipVertex()]{\mbox{$\IsFlipVertex(n,x)$}}
\label{alg:isflipvertex}
\vspace{.2em}
\KwIn{An integer $n\geq 1$, a vertex $x\in D_{2n}^{=0}(n)$} 
\KwOut{$\{\true,\false\}$}
\vspace{.2em}
$T:=h^{-1}(x)$ \tcc*{compute corresponding rooted tree} \label{line:h-inverse}
\KwIf $T\in\cT_{n,1}^*$ \KwThen \Return $\IsFlipTree_1(T)$ \label{line:tau_1-preimage} \;
\KwElseIf $T\in\tau_1(\cT_{n,1}^*)$ \KwThen \Return $\IsFlipTree_1(\tau_1^{-1}(T))$ \label{line:tau_1-image} \;
\KwElseIf $T\in\cT_{n,2}^*$ \KwThen \Return $\IsFlipTree_2(T)$ \label{line:tau_2-preimage} \;
\KwElseIf $T\in\tau_2(\cT_{n,2}^*)$ \KwThen \Return $\IsFlipTree_2(\tau_2^{-1}(T))$ \label{line:tau_2-image} \;
\KwElse \Return $\false$ \;
\end{algorithm}
%\decmargin{1em}

We proceed to define the auxiliary functions $\IsFlipTree_1()$ and $\IsFlipTree_2()$ that have return values $\{\true,\false\}$.
For brevity we define them in textual form, not as pseudocode. 
For the definitions we need to introduce the notion of tree rotation.

\textit{Rotation $\trot()$ of ordered rooted trees.}
For an ordered rooted tree $T\in\cT_n^*$ we let $\trot(T)$ denote the tree that is obtained from $T$ by shifting the root to the leftmost child of the root. 
In terms of lattice paths, considering the unique decomposition $T=1\circ T_\ell\circ 0\circ T_r$ with $T_\ell\in D_{2k}^{=0}(k)$ for some $k\geq 0$, we have $\trot(T)=T_\ell\circ 1\circ T_r\circ 0$.
Note that we obtain the same tree again after $2n/s$ such rotations, where the value of $s$ depends on the symmetries of the tree.
As an example, the top part of Figure~\ref{fig:rot-trees} shows the effect of applying the rotation operation to all 14 trees in the set $\cT_4^*$.

\textit{The function $\IsFlipTree_1(T)$:} Given a tree $T\in\cT_{n,1}^*$, compute all rotated versions of it that are also in the set $\cT_{n,1}^*$. 
Return $\true$ if $T$ is lexicographically smallest among them, and $\false$ otherwise.

\textit{The function $\IsFlipTree_2(T)$:} Given a tree $T\in\cT_{n,2}^*$, if $T=1^{n-1}\circ 0^{n-2}\circ 100$ or if $T$ has more than one thin leaf, return $\false$. 
Otherwise, let $v$ be the unique thin leaf of $T$, $v'$ the parent of $v$, $v''$ the parent of $v'$, and $w$ the clockwise-next leaf of $v$, and let $T'$ be the tree obtained from $T$ by replacing the edge $(v,v')$ by $(v,v'')$ so that $v$ becomes a child of $v''$ to the left of $v'$ (so $T'$ has only thick leaves), and by rotating it such that the leaf $w$ becomes the root.
Let $d$ be the distance between the root $w$ and the leftmost leaf $v$ of $T'$. 
Compute all other rotated versions of $T'$ for which the root is a leaf and the parent of the leftmost leaf $x$ has another leaf as its child to the right of $x$ (initially, $v$ and $v'$ are these two leaves). 
For each of them compute the distance $d'$ between the root and the leftmost leaf.
Return $\true$ if $d'\leq d$ for all of them and if $T'$ is lexicographically smallest among all rotated trees with $d'=d$, and $\false$ otherwise.

When referring to lexicographic comparisons between ordered rooted trees in the previous definitions we mean lexicographic comparisons between the corresponding bitstrings.

\section{Correctness of the algorithm}
\label{sec:correctness}

In this section we prove the correctness of the algorithm presented in the previous section.
Before doing this, we recapitulate the inductive construction of paths $\cP_n(k,k+1)$ and $\tcP_n(k,k+1)$ described in \cite{MR3483129} in Sections~\ref{sec:recap-paths} and \ref{sec:recap-flippable} below.
These sections are essentially duplicate material from \cite{MR3483129}, and readers familiar with this construction may want to skip them.
The main goal of Section~\ref{sec:correctness-paths} is to prove that those paths are the same as the ones produced by our recursive algorithm $\Paths()$, so this will justify denoting both with the same symbol $\cP_n(k,k+1)$ or $\tcP_n(k,k+1)$, respectively.
Once this is established, we apply several lemmas from \cite{MR3483129} to prove the correctness of the algorithm $\HamCycle()$ in Section~\ref{sec:correctness-hamcycle}.
When looking up those lemmas the reader should be aware of some minor differences in notation between \cite{MR3483129} and our paper:
We use bitstrings to also denote the corresponding lattice paths and ordered rooted trees, implicitly using the corresponding bijections.
In the paper \cite{MR3483129}, those bijections are made explicit and denoted $\varphi$ and $\psi$, and upsteps and downsteps of a lattice path are denoted by $\upstep$ and $\downstep$ rather than 1 and 0.

\subsection{Recap: Inductive construction of paths from \texorpdfstring{\cite{MR3483129}}{[M\"ut16]}}
\label{sec:recap-paths}

The construction is parametrized by some sequence $(\alpha_{2i})_{i\geq 1}$, $\alpha_{2i}\in\{0,1\}^{i-1}$.
Given this sequence, we inductively construct a set $\cP_{2n}(k,k+1)$ of disjoint oriented paths in $Q_{2n}(k,k+1)$ that start and end in $B_{2n}(k)$ for all $n\geq 1$ and all $k=n,n+1,\ldots,2n-1$ such that the following conditions hold:
\begin{enumerate}[label=(\alph*),topsep=0mm,leftmargin=7mm]
\item The paths in $\cP_{2n}(n,n+1)$ visit all vertices in the sets $B_{2n}(n+1)$ and $B_{2n}(n)$.
\item For $k=n+1,\ldots,2n-1$, the paths in $\cP_{2n}(k,k+1)$ visit all vertices in the set $B_{2n}(k+1)$, and the only vertices not visited in the set $B_{2n}(k)$ are exactly the elements in the set $S(\cP_{2n}(k-1,k))$.
\end{enumerate}

\textbf{Induction basis $n=1$ ($Q_2$):}
The induction basis is the set of paths $\cP_2(1,2):=\{(10,11,01)\}$, consisting only of a single oriented path on three vertices.
It is easily checked that this set of paths satisfies properties~(a) and (b) (property~(b) is satisfied trivially).

\textbf{Induction step $n\rightarrow n+1$ ($Q_{2n}\rightarrow Q_{2n+2}$), $n\geq 1$:}
The inductive construction consists of two intermediate steps.

\textit{First intermediate step: Construction of a 2-factor in the middle levels graph $Q_{2n+1}(n,n+1)$.}
Using only the paths in the set $\cP_{2n}(n,n+1)$ and the parameter $\alpha_{2n}=(\alpha_{2n}(1),\ldots,\alpha_{2n}(n-1))\in\{0,1\}^{n-1}$ we first construct a 2-factor in the middle levels graph $Q_{2n+1}(n,n+1)$.

The parameter $\alpha_{2i}$ is used to generalize the permutation $\pi$ introduced before:
Specifically, we let $\pi_{\alpha_{2n}}$ denote the permutation on the set $\{0,1\}^{2n}$ that swaps any two adjacent bits at positions $2i$ and $2i+1$ for all $i=1,2,\ldots,n-1$, if and only if $\alpha_{2n}(i)=1$, and that leaves the bits at position $1$ and $2n$ unchanged.
Note that if $\alpha_{2n}=(1,1,\ldots,1)$, then $\pi_{\alpha_{2n}}$ equals the mapping $\pi$ from before.
On the other hand, if $\alpha_{2n}=(0,0,\ldots,0)$, then no bits are swapped and $\pi_{\alpha_{2n}}=\id$ is simply the identity mapping.
For any bitstring $x\in \{0,1\}^{2n}$ we then define
\begin{equation} \label{eq:f-alpha}
  f_{\alpha_{2n}}(x):=\ol{\rev}(\pi_{\alpha_{2n}}(x)) \enspace.
\end{equation}
Note that we have $f_{\alpha_{2n}}=\ol{\rev}\bullet \pi$ for $\alpha_{2n}=(1,1,\ldots,1)$ and $f_{\alpha_{2n}}=\ol{\rev}$ for $\alpha_{2n}=(0,0,\ldots,0)$.

It was proved in \cite[Lemma~3]{MR3483129} that for any $\alpha_{2n}\in\{0,1\}^{n-1}$ the mapping $f_{\alpha_{2n}}$ maps each of the sets $F(\cP_{2n}(n,n+1))$ and $L(\cP_{2n}(n,n+1))$ onto itself.
As explained in Section~\ref{sec:hamcyclenext}, the middle layer graph $Q_{2n+1}(n,n+1)$ can be decomposed into the graphs $Q_{2n}(n,n+1)\circ 0$ and $Q_{2n}(n-1,n)\circ 1$ plus the edges from the matching $M_{2n+1}=\{(x\circ 0 ,x\circ 1 )\mid x\in B_{2n}(n)\}$ along which the last bit is flipped.
Denoting by $M_{2n+1}^{FL}$ the edges from $M_{2n+1}$ that have one end vertex in the set $\big(F(\cP_{2n}(n,n+1))\cup L(\cP_{2n}(n,n+1))\big)\circ 0\seq B_{2n}(n)\circ 0$ (and the other in the set $\big(F(\cP_{2n}(n,n+1))\cup L(\cP_{2n}(n,n+1))\big)\circ 1\seq B_{2n}(n)\circ 1$), the union
\begin{equation} \label{eq:2-factor-alpha}
  \cC_{2n+1}:=\cP_{2n}(n,n+1)\circ 0\cup f_{\alpha_{2n}}(\cP_{2n}(n,n+1))\circ 1\cup M_{2n+1}^{FL}
\end{equation}
therefore yields a 2-factor in the middle levels graph, with the property that on every cycle of $\cC_{2n+1}$, every edge of the form $(F(P),S(P))\circ 0$ for some $P\in\cP_{2n}(n,n+1)$ is oriented the same way.
Note that \eqref{eq:2-factor} is the special case of \eqref{eq:2-factor-alpha} with $\alpha_{2n}=(0,0,\ldots,0)$.
In the correctness proof for our algorithm we will later also consider the case $\alpha_{2n}=(1,1,\ldots,1)$.
Even though we are eventually only interested in the 2-factor $\cC_{2n+1}$ defined in \eqref{eq:2-factor-alpha}, we need to specify how to proceed with the inductive construction of the sets of paths $\cP_{2n+2}(k,k+1)$.

\textit{Second intermediate step: Splitting up the 2-factor into paths.}
We proceed by describing how the sets of paths $\cP_{2n+2}(k,k+1)$ for all $k=n+1,n+2,\ldots,2n+1$ satisfying properties~(a) and (b) are defined, using the previously constructed sets $\cP_{2n}(k,k+1)$ and the 2-factor $\cC_{2n+1}$ defined in the first intermediate step.

Consider the decomposition of $Q_{2n+2}$ into $Q_{2n}\circ 00$, $Q_{2n}\circ 10$, $Q_{2n}\circ 01$ and $Q_{2n}\circ 11$ plus the two perfect matchings $M_{2n+2}$ and $M_{2n+2}'$ along which the last and second to last bit are flipped, respectively.
For all $k=n+2,\ldots,2n+1$ we define
\begin{equation}
\label{eq:ind-step1-P}
\begin{split}
  \cP_{2n+2}(k,k+1) &:= \cP_{2n}(k,k+1)\circ 00\cup\cP_{2n}(k-1,k)\circ 10 \\
                    &\qquad \cup \cP_{2n}(k-1,k)\circ 01\cup\cP_{2n}(k-2,k-1)\circ 11  \enspace,
\end{split}
\end{equation}
where we use the convention $\cP_{2n}(2n,2n+1):=\emptyset$ and $\cP_{2n}(2n+1,2n+2):=\emptyset$ to unify treatment of the sets of paths $\cP_{2n+2}(2n,2n+1)$ and $\cP_{2n+2}(2n+1,2n+2)$ between the uppermost levels of $Q_{2n+2}$.
Note that so far none of the edges from the matchings $M_{2n+2}$ or $M_{2n+2}'$ is used.

The definition of the set $\cP_{2n+2}(n+1,n+2)$ is slightly more involved.
Note that the graph $Q_{2n+2}(n+1,n+2)$ can be decomposed into $Q_{2n+1}(n+1,n+2)\circ 0$ and $Q_{2n+1}(n,n+1)\circ 1$ plus the edges from $M_{2n+2}$ that connect the vertices in the set $B_{2n+1}(n+1)\circ 0$ to the vertices in the set $B_{2n+1}(n+1)\circ 1$.
The first graph can be further decomposed into $Q_{2n}(n+1,n+2)\circ 00$ and $Q_{2n}(n,n+1)\circ 10$ plus some matching edges that are not relevant here.
The second graph is the middle levels graph of $Q_{2n+1}\circ 1$.
Let $\cC_{2n+1}^-$ denote the graph obtained from the 2-factor $\cC_{2n+1}$ defined in \eqref{eq:2-factor-alpha} by removing every edge of the form $(F(P),S(P))\circ 0$ for some $P\in\cP_{2n}(n,n+1)$.
As on every cycle of $\cC_{2n+1}$ every such edge is oriented the same way, $\cC_{2n+1}^-$ is a set of paths (visiting all vertices of the middle levels graph $Q_{2n+1}(n,n+1)$), with the property that each of those paths starts at a vertex of the form $S(P)\circ 0$ and ends at a vertex of the form $F(P')\circ 0$ for two paths $P,P'\in\cP_{2n}(n,n+1)$.
Denoting by $M_{2n+2}^S$ the edges from $M_{2n+2}$ that have one end vertex in the set $S(\cP_{2n}(n,n+1))\circ 00\seq B_{2n}(n+1)\circ 00$ (and the other in the set $S(\cP_{2n}(n,n+1))\circ 01\seq B_{2n}(n+1)\circ 01$), it follows that
\begin{equation}
\label{eq:new-paths}
  \cP_{2n+2}^*:=M_{2n+2}^S\cup \cC_{2n+1}^-\circ 1
\end{equation}
is a set of oriented paths that start and end in $B_{2n+2}(n+1)$, where we choose the orientation of each path such that the edge from the set $M_{2n+2}^S$ is the first edge.
Note that we have
\begin{equation}
\label{eq:new-paths-F}
  F(\cP_{2n+2}^*)=S(\cP_{2n}(n,n+1))\circ 00 \enspace.
\end{equation}
We then define
\begin{equation}
\label{eq:ind-step2-P}
  \cP_{2n+2}(n+1,n+2):=\cP_{2n}(n+1,n+2)\circ 00\cup\cP_{2n}(n,n+1)\circ 10\cup \cP_{2n+2}^* \enspace,
\end{equation}
where in the case $n=1$ we use the convention $\cP_2(2,3):=\emptyset$.

It was argued in \cite[Section~2.2]{MR3483129} that the sets of paths $\cP_{2n+2}(k,k+1)$, $k=n+1,n+2,\ldots,2n+1$, defined in \eqref{eq:ind-step1-P} and \eqref{eq:ind-step2-P} have properties~(a) and (b).
We omit those arguments here.

\subsection{Recap: Inductive construction of flippable pairs from \texorpdfstring{\cite{MR3483129}}{[M\"ut16]}}
\label{sec:recap-flippable}

Let $\cP_{2n}(k,k+1)$ be the sets of paths defined in Section~\ref{sec:recap-paths} for an arbitrary parameter sequence $(\alpha_{2i})_{i\geq 1}$, $\alpha_{2i}\in\{0,1\}^{i-1}$.
In the following we show how to inductively construct a set of flippable pairs $\cX_{2n}(k,k+1)$ for the set $\cP_{2n}(k,k+1)$ for all $n\geq 2$ and all $k=n,n+1,\ldots,2n-1$, and the corresponding flipped paths $\tcP_{2n}(k,k+1)$.
So the set $\cX_{2n}(k,k+1)$ contains pairs of paths $(P,P')$, $P,P'\in \cP_{2n}(k,k+1)$, that form a flippable pair, where every path appears in at most one such pair, and the set $\tcP_{2n}(k,k+1)$ is the union of all paths $R,R'$ that form a corresponding flipped pair $(R,R')$ for $(P,P')$.
This construction arises very naturally from the inductive construction of the sets $\cP_{2n}(k,k+1)$ given in the previous section.

\textbf{Induction basis $n=2$ ($Q_4$):}
It is easy to check that the set of paths $\cP_4(2,3)$ arising from the inductive construction described in the previous section after one step satisfies $\cP_4(2,3)=\{P,P'\}$ with $P$ and $P'$ as defined in \eqref{eq:P423} (as $\alpha_2=()$ this step does not involve any parameter choices yet).
The set $\cX_4(2,3):=\{(P,P')\}$ is a set of flippable pairs for $\cP_4(2,3)$, which can be seen by considering the flipped paths $(R,R')$ with $R$ and $R'$ as defined in \eqref{eq:tP423}.
For completeness we also define $\cX_4(3,4):=\emptyset$, which is trivially a set of flippable pairs for the set of paths $\cP_4(3,4)$ constructed before.
The corresponding flipped paths are therefore $\tcP_4(2,3):= \{R,R'\}$ and $\tcP_4(3,4):=\emptyset$.

\textbf{Induction step $n\rightarrow n+1$ ($Q_{2n}\rightarrow Q_{2n+2}$), $n\geq 2$:}
Consider the sets of flippable pairs $\cX_{2n}(k,k+1)$, $k=n,n+1,\ldots,2n-1$, for the sets of paths $\cP_{2n}(k,k+1)$, and the corresponding flipped paths $\tcP_{2n}(k,k+1)$.
In the following we describe how to use them to construct sets of flippable pairs $\cX_{2n+2}(k,k+1)$, $k=n+1,n+2,\ldots,2n+1$, for the sets $\cP_{2n+2}(k,k+1)$ in $Q_{2n+2}(k,k+1)$, and corresponding flipped paths $\tcP_{2n}(k,k+1)$.

For all $k=n+2,\ldots,2n+1$ we define, in analogy to \eqref{eq:ind-step1-P},
\begin{equation}
\label{eq:ind-step1-flipp}
\begin{split}
  \cX_{2n+2}(k,k+1) &:= \cX_{2n}(k,k+1)\circ 00\cup\cX_{2n}(k-1,k)\circ 10 \\
                    &\qquad \cup \cX_{2n}(k-1,k)\circ 01\cup\cX_{2n}(k-2,k-1)\circ 11 \enspace,
\end{split}
\end{equation}
where we use the convention $\cX_{2n}(2n,2n+1):=\emptyset$ and $\cX_{2n}(2n+1,2n+2):=\emptyset$ to unify treatment of the sets of flippable pairs $\cX_{2n+2}(2n,2n+1)$ and $\cX_{2n+2}(2n+1,2n+2)$ between the uppermost levels of $Q_{2n+2}$.
The sets of flippable pairs on the right hand side of \eqref{eq:ind-step1-flipp} clearly lie in four disjoint subgraphs of $Q_{2n+2}(k,k+1)$, so by induction $\cX_{2n+2}(k,k+1)$ is indeed a set of flippable pairs for $\cP_{2n+2}(k,k+1)$.

The corresponding flipped paths are defined as
\begin{equation*}
\begin{split}
  \tcP_{2n+2}(k,k+1) &:= \tcP_{2n}(k,k+1)\circ 00\cup\tcP_{2n}(k-1,k)\circ 10 \\
                     &\qquad \cup \tcP_{2n}(k-1,k)\circ 01\cup\tcP_{2n}(k-2,k-1)\circ 11 \enspace,
\end{split}
\end{equation*}
where we use the convention $\tcP_{2n}(2n,2n+1):=\emptyset$ and $\tcP_{2n}(2n+1,2n+2):=\emptyset$.

To define the set $\cX_{2n+2}(n+1,n+2)$, we consider the oriented paths $\cP_{2n+2}^*$ defined in \eqref{eq:new-paths} (recall that these paths originate from splitting up the 2-factor $\cC_{2n+1}$ defined in \eqref{eq:2-factor-alpha}).
By \eqref{eq:2-factor-alpha} and \eqref{eq:new-paths}, every path $P^+\in\cP_{2n+2}^*$ has the following structure:
There are two paths $P,\Phat\in\cP_{2n}(n,n+1)$ with $f_{\alpha_{2n}}(L(\Phat))=L(P)$ such that $P^+$ contains all edges except the first one from $P\circ 01$ and all edges from $f_{\alpha_{2n}}(\Phat)\circ 11$ ($P^+$ has three more edges, two from the matching $M_{2n+1}^{FL}\circ 1$ and one from the matching $M_{2n+2}^S$).
By this structural property of paths from $\cP_{2n+2}^*$ and the fact that $f_{\alpha_{2n}}$ is an isomorphism between the graphs $Q_{2n}(n,n+1)$ and $Q_{2n}(n-1,n)$, the set
\begin{equation}
\label{eq:new-flipp-paths}
\begin{split}
  \cX_{2n+2}^* &:= \big\{(P^+,P^{+'}) \mid P^+,P^{+'}\in\cP_{2n+2}^* \text{ and there is a flippable pair } (\Phat,\Phat')\in\cX_{2n}(n,n+1) \\
               &\hspace{3cm} \text{ with } f_{\alpha_{2n}}(\Phat)\circ 11 \seq P^+ \text{ and } f_{\alpha_{2n}}(\Phat')\circ 11 \seq P^{+'} \big\} \enspace.
\end{split}
\end{equation}
is a set of flippable pairs for $\cP_{2n+2}^*$.
We now define, in analogy to \eqref{eq:ind-step2-P},
\begin{equation}
\label{eq:ind-step2-flipp}
  \cX_{2n+2}(n+1,n+2):=\cX_{2n}(n+1,n+2)\circ 00\cup\cX_{2n}(n,n+1)\circ 10\cup \cX_{2n+2}^* \enspace.
\end{equation}
The sets of flippable pairs on the right hand side of \eqref{eq:ind-step2-flipp} lie in three disjoint subgraphs of $Q_{2n+2}(n+1,n+2)$, so by induction $\cX_{2n+2}(n+1,n+2)$ is indeed a set of flippable pairs for $\cP_{2n+2}(n+1,n+2)$.

The corresponding flipped paths are defined as
\begin{equation}
\label{eq:ind-step2-flipped}
  \tcP_{2n+2}(n+1,n+2):=\tcP_{2n}(n+1,n+2)\circ 00\cup\tcP_{2n}(n,n+1)\circ 10\cup \tcP_{2n+2}^* \enspace,
\end{equation}
where $\tcP_{2n+2}^*$ is the union of all paths obtained by considering every flippable pair $(P^+,P^{+'})$ in the set $\cX_{2n+2}^*$ defined in \eqref{eq:new-flipp-paths} and by modifying $P^+$ and $P^{+'}$ as follows:
Let $(\Phat,\Phat')\in\cX_{2n}(n,n+1)$ be such that $f_{\alpha_{2n}}(\Phat)\circ 11 \seq P^+$ and $f_{\alpha_{2n}}(\Phat')\circ 11 \seq P^{+'}$ (see the right hand side of \eqref{eq:new-flipp-paths}) and let $\Pbar$ and $\Pbar'$ be the corresponding flipped paths from $\tcP_{2n}(n,n+1)$ so that $F(\Phat)=F(\Pbar)$, $F(\Phat')=F(\Pbar')$, $L(\Phat)=L(\Pbar')$, $L(\Phat')=L(\Pbar)$.
We replace the subpath $f_{\alpha_{2n}}(\Phat)\circ 11$ of $P^+$ by the subpath $f_{\alpha_{2n}}(\Pbar')\circ 11$ and the last edge $(F(\Phat)\circ 11,F(\Phat)\circ 01)$ of $P^+$ by the edge $(F(\Pbar')\circ 11,F(\Pbar')\circ 01)$. 
Similarly, we replace the subpath $f_{\alpha_{2n}}(\Phat')\circ 11$ of $P^{+'}$ by the subpath $f_{\alpha_{2n}}(\Pbar)\circ 11$ and the last edge $(F(\Phat')\circ 11,F(\Phat')\circ 01)$ of $P^{+'}$ by the edge $(F(\Pbar)\circ 11,F(\Pbar)\circ 01)$.
Clearly, the resulting paths are a flipped pair of paths corresponding to the flippable pair of paths $(P^+,P^{+'})$ that was added to $\cX_{2n+2}^*$.

\subsection{Correctness of the algorithm \texorpdfstring{$\Paths()$}{Paths()}}
\label{sec:correctness-paths}

In this section we show that the paths $\cP_{2n}(k,k+1)$ and $\tcP_{2n}(k,k+1)$ in the graph $Q_{2n}(k,k+1)$ computed by the algorithm $\Paths()$ that form the basic building blocks of our Hamilton cycle algorithm (recall \eqref{eq:2-factor}) are exactly the same as the ones defined in Sections~\ref{sec:recap-paths} and \ref{sec:recap-flippable}.
We shall see that essentially, the algorithm $\Paths()$ is a recursive reformulation of this inductive construction.
This allows us to exploit all the properties proved in \cite{MR3483129} about these paths.
In particular, we establish that the paths $\cP_{2n}(n,n+1)$ satisfy properties~(i) and (ii) mentioned in Section~\ref{sec:paths}.
Recall that all paths we consider are oriented from the first vertex to the last, and the notions of previous and next vertices on a path are defined with respect to this orientation.

\begin{lemma}
\label{lemma:paths-original}
Let $n\geq 1$ and $k\in\{n,n+1,\ldots,2n-1\}$ be fixed, and let $\cP_{2n}(k,k+1)$ be the set of paths in $Q_{2n}(k,k+1)$ defined in Section~\ref{sec:recap-paths} for the parameter sequence $\alpha_{2i}=(1,1,\ldots,1)\in \{0,1\}^{i-1}$, $i=1,2,\ldots,n-1$.
For any path $P\in\cP_{2n}(k,k+1)$ and any two consecutive vertices $u$ and $v$ on $P$ (i.e., $P$ has the form $P=(\ldots,u,v,\ldots)$) we have $v=\Paths(n,k,u,\nextv,\false)$ and $u=\Paths(n,k,v,\prevv,\false)$.
\end{lemma}

In the following we will repeatedly use that by \cite[Lemma~11]{MR3483129} the sets of first, second and last vertices of the paths in $\cP_{2n}(k,k+1)$ satisfy the relations
\begin{subequations}
\label{eq:FSLP}
\begin{align}
  F(\cP_{2n}(k,k+1)) &= D_{2n}^{=0}(k) \enspace, \label{eq:FP} \\
  S(\cP_{2n}(k,k+1)) &= D_{2n}^{>0}(k+1) \enspace, \label{eq:SP} \\
  L(\cP_{2n}(k,k+1)) &= D_{2n}^-(k) \enspace. \label{eq:LP}
\end{align}
\end{subequations}

\begin{proof}
Note that when calling the algorithm $\Paths()$ with the parameter $\flip=\false$, this parameter retains its value in all recursive calls, so we can ignore lines~\ref{line:n2}--\ref{line:base-case2} of the algorithm for this proof.
We prove the lemma by induction on $n$.
To settle the base case $n=1$ observe that the set of paths $\cP_2(1,2)$ defined in Section~\ref{sec:recap-paths} and containing only a single path is the same as the one defined in \eqref{eq:P212} and returned in lines~\ref{line:n1}--\ref{line:base-case1}.

For the induction step $n\rightarrow n+1$ we assume that the lemma holds for $n\geq 1$ and all $k=n,n+1,\ldots,2n-1$, and prove it for $n+1$ and all $k=n+1,n+2,\ldots,2n+1$.
We distinguish the cases $k=n+1$ and $k\in\{n+2,\ldots,2n+1\}$.

For the case $k\in\{n+2,\ldots,2n+1\}$ let $\cP_{2n+2}(k,k+1)$ be the set of paths defined in \eqref{eq:ind-step1-P}, let $P$ be a path from this set and let $u$ and $v$ be two consecutive vertices on $P$.
Note that for both calls $\Paths(n+1,k,u,\nextv,\false)$ and $\Paths(n+1,k,v,\prevv,\false)$ the condition in line~\ref{line:upper-layers-cond} is satisfied, entailing the recursive call in line~\ref{line:recurse-upper-layers}.
The claim therefore follows from the definition \eqref{eq:ind-step1-P} and by induction: The quantity $k-u_{2n}-u_{2n+1}$ computed in line~\ref{line:recurse-upper-layers} evaluates to $k$, $k-1$ or $k-2$, depending on the values of the last two bits of $u$ and an analogous statement holds for $v$.

For the case $k=n+1$ let $\cP_{2n+2}(n+1,n+2)$ be the set of paths defined in \eqref{eq:ind-step2-P}, let $P$ be a path from this set and let $u$ and $v$ be two consecutive vertices on $P$.
We distinguish three cases depending on which of the three sets on the right hand side of \eqref{eq:ind-step2-P} the path $P$ is contained in.

Case 1: If $P$ is contained in the second set on the right hand side of \eqref{eq:ind-step2-P}, i.e., in the set $\cP_{2n}(n,n+1)\circ 10$, then the last two bits of $u$ and $v$ are equal to $10$, so for both calls $\Paths(n+1,k,u,\nextv,\false)$ and $\Paths(n+1,k,v,\prevv,\false)$ the condition in line~\ref{line:xp10} is satisfied, entailing the recursive call in line~\ref{line:recurse-10}.
In this case the claim follows by induction.

Case 2: If $P$ is contained in the first set on the right hand side of \eqref{eq:ind-step2-P}, i.e., in the set $\cP_{2n}(n+1,n+2)\circ 00$, then the last two bits of $u$ and $v$ are equal to $00$, so for both calls $\Paths(n+1,k,u,\nextv,\false)$ and $\Paths(n+1,k,v,\prevv,\false)$ the condition in line~\ref{line:xp00} is satisfied. 
However, neither $u$ nor $v$ satisfies the condition in line~\ref{line:xp00f}, as by \eqref{eq:new-paths-F} and \eqref{eq:SP}, the vertices in $D_{2n}^{>0}(n+1)\circ 00$ are all visited by the paths $\cP_{2n+2}^*$ that form the third set on the right hand side of set of \eqref{eq:ind-step2-P}, so $P$ does not contain any of these vertices. 
Consequently, both cases entail the recursive call in line~\ref{line:recurse-00} and the claim follows by induction.

Case 3: If $P$ is contained in the third set on the right hand side of \eqref{eq:ind-step2-P}, i.e., in the set $\cP_{2n+2}^*$, then by the definitions \eqref{eq:2-factor-alpha} and \eqref{eq:new-paths} and the fact that for $\alpha_{2n}=(1,1,\ldots,1)\in \{0,1\}^{n-1}$ the mapping defined in \eqref{eq:f-alpha} satisfies $f_{\alpha_{2n}}=\ol{\rev}\bullet \pi$, the sequence of edges of the path $P$ when traversing it from its first to its last vertex has the form $(e_1,E_2,e_3,E_4,e_5)$, where $e_1$, $e_3$ and $e_5$ are single edges and $E_2$ and $E_4$ are sequences of edges that satisfy the following conditions:
There are two paths $P',P''\in\cP_{2n}(n,n+1)$ such that
\begin{align*}
  e_1 &= \big(S(P')\circ 00,S(P')\circ 01\big) \enspace, \\
  e_3 &= \big(L(P')\circ 01,L(P')\circ 11\big)=\big(L(P')\circ 01,\ol{\rev}(\pi(L(P'')))\circ 11\big) \enspace, \\
  e_5 &= \big(\ol{\rev}(\pi(F(P'')))\circ 11,\ol{\rev}(\pi(F(P'')))\circ 01\big) \enspace,
\end{align*}
$E_2$ is given by traversing the edges of $P'\circ 01$ starting at the vertex $S(P')\circ 01$ and ending at the vertex $L(P')\circ 01$, and $E_4$ is given by the traversing all the edges of $\ol{\rev}(\pi(P''))\circ 11$ in reverse order starting at the vertex $\ol{\rev}(\pi(L(P'')))\circ 11$ and ending at the vertex $\ol{\rev}(\pi(F(P'')))\circ 11$.
We distinguish five subcases depending on which part of $P$ the edge $(u,v)$ belongs to.

Case~3.1: If $(u,v)=e_1=\big(S(P')\circ 00,S(P')\circ 01\big)$, then the vertices $u$ and $v$ differ in the last bit. 
It follows from \eqref{eq:SP} that for the call $\Paths(n+1,k,u,\nextv,\false)$ the conditions in line~\ref{line:xp00} and \ref{line:xp00f} are both satisfied, so the algorithm correctly returns $v$ which is obtained from $u$ by flipping the last bit.
Similarly, for the call $\Paths(n+1,k,v,\prevv,\false)$ the conditions in line~\ref{line:xp01} and \ref{line:xp01f3} are both satisfied, so the algorithm correctly returns $u$ which is obtained from $v$ by flipping the last bit.

Case~3.2: If $(u,v)$ is an edge from $E_2$, then the last two bits of $u$ and $v$ are equal to $01$, so for both calls $\Paths(n+1,k,u,\nextv,\false)$ and $\Paths(n+1,k,v,\prevv,\false)$ the condition in line~\ref{line:xp01} is satisfied. 
However, the conditions in line~\ref{line:xp01f1}, \ref{line:xp01f2} and \ref{line:xp01f3} are not satisfied:
To see this recall \eqref{eq:FP}, \eqref{eq:SP} and \eqref{eq:LP}, and that by the definition of $E_2$, $u$ and $v$ are different from $F(P')\circ 01$, $u$ is different from $L(P')\circ 01$, and $v$ is different from $S(P')\circ 01$.
Consequently, both cases entail the recursive call in line~\ref{line:recurse-01} and the claim follows by induction.

Case~3.3: If $(u,v)=e_3=\big(L(P')\circ 01,L(P')\circ 11\big)$, then the vertices $u$ and $v$ differ in the second to last bit. 
It follows from \eqref{eq:LP} that for the call $\Paths(n+1,k,u,\nextv,\false)$ the conditions in line~\ref{line:xp01} and \ref{line:xp01f2} are both satisfied, so the algorithm correctly returns $v$ which is obtained from $u$ by flipping the second to last bit.
Similarly, for the call $\Paths(n+1,k,v,\prevv,\false)$ the conditions in line~\ref{line:xp11} and \ref{line:xp11f2} are both satisfied, so the algorithm correctly returns $u$ which is obtained from $v$ by flipping the second to last bit.

Case~3.4: If $(u,v)$ is an edge from $E_4$, then the last two bits of $u$ and $v$ are equal to $11$, so for both calls $\Paths(n+1,k,u,\nextv,\false)$ and $\Paths(n+1,k,v,\prevv,\false)$ the condition in line~\ref{line:xp11} is satisfied. 
However, the conditions in line~\ref{line:xp11f1} and \ref{line:xp11f2} are not satisfied:
To see this recall \eqref{eq:FP} and \eqref{eq:LP} and that the mapping $\ol{\rev}\bullet \pi$ maps each of the sets $F(\cP_{2n}(n,n+1))=D_{2n}^{=0}(n)$ and $L(\cP_{2n}(n,n+1))=D_{2n}^-(n)$ onto itself, and note that by the definition of $E_4$, $u$ is different from $\ol{\rev}(\pi(F(P'')))\circ 11$ and $v$ is different from $\ol{\rev}(\pi(L(P'')))\circ 11$.
Consequently, both cases entail the recursive call in line~\ref{line:recurse-11}. 
The claim follows by induction, noting that inverting the value of the variable $\dir$ accounts for the fact that in $E_4$ the edges of $\ol{\rev}(\pi(P''))\circ 11$ are traversed in reverse order.

Case~3.5: If $(u,v)=e_5=\big(\ol{\rev}(\pi(F(P'')))\circ 11,\ol{\rev}(\pi(F(P'')))\circ 01\big)$, then the vertices $u$ and $v$ differ in the second to last bit. 
It follows from \eqref{eq:FP} that $\ol{\rev}(\pi(F(P'')))\in D_{2n}^{=0}(n)$, so for the call $\Paths(n+1,k,u,\nextv,\false)$ the conditions in line~\ref{line:xp11} and \ref{line:xp11f1} are both satisfied and the algorithm correctly returns $v$ which is obtained from $u$ by flipping the second to last bit.
Similarly, for the call $\Paths(n+1,k,v,\prevv,\false)$ the conditions in line~\ref{line:xp01} and \ref{line:xp01f1} are both satisfied, so the algorithm correctly returns $u$ which is obtained from $v$ by flipping the second to last bit.

This completes the proof.
\end{proof}

\begin{lemma}
\label{lemma:paths-flipped}
Let $n\geq 2$ and $k\in\{n,n+1,\ldots,2n-1\}$ be fixed, let $\cP_{2n}(k,k+1)$ be the set of paths defined in Section~\ref{sec:recap-paths} and $\cX_{2n}(k,k+1)$ the set of flippable pairs for $\cP_{2n}(k,k+1)$ and $\tcP_{2n}(k,k+1)$ the corresponding set of flipped paths in $Q_{2n}(k,k+1)$ defined in Section~\ref{sec:recap-flippable} for the parameter sequence $\alpha_{2i}=(1,1,\ldots,1)\in \{0,1\}^{i-1}$, $i=1,2,\ldots,n-1$.
For any path $P\in\tcP_{2n}(k,k+1)$ and any two consecutive vertices $u$ and $v$ on $P$ (i.e., $P$ has the form $P=(\ldots,u,v,\ldots)$) we have $v=\Paths(n,k,u,\nextv,\true)$ and $u=\Paths(n,k,v,\prevv,\true)$.
\end{lemma}

From \eqref{eq:FSLP} and from the definition of flippable/flipped pairs of paths we obtain that
\begin{subequations}
\begin{align}
  F(\tcP_{2n}(k,k+1)) &= D_{2n}^{=0}(k) \enspace, \label{eq:FtP} \\
  L(\tcP_{2n}(k,k+1)) &= D_{2n}^-(k) \enspace. \label{eq:LtP}
\end{align}
\end{subequations}
\TM{Equality for second vertices does ***not*** hold.}

\begin{proof}
We prove the lemma by induction on $n$, in an analogous way as the proof of Lemma~\ref{lemma:paths-original}.
One crucial difference, however, is the following: When calling the algorithm $\Paths()$ with the parameter $\flip=\true$, this parameter retains its value in all recursive calls except in line~\ref{line:recurse-01}, where it is set to $\false$.

To settle the base case $n=2$ observe that the set of paths $\tcP_4(2,3)$ defined in Section~\ref{sec:recap-flippable} (consisting of two paths on three and seven vertices, respectively) is the same as the one defined in \eqref{eq:tP423} and returned in lines~\ref{line:n2}--\ref{line:base-case2}. 
The set of paths $\tcP_4(3,4)$ defined in Section~\ref{sec:recap-flippable} is empty, so the claim is trivially true.

The induction step $n\rightarrow n+1$ proceeds in an analogous way as in the proof of Lemma~\ref{lemma:paths-original}.
We only give the details for case~3, which is proved in a slightly different way:

Case~3': If $P$ is contained in the third set on the right hand side of \eqref{eq:ind-step2-flipped}, i.e., in the set $\tcP_{2n+2}^*$, then by the definition given after \eqref{eq:ind-step2-flipped} and the fact that for $\alpha_{2n}=(1,1,\ldots,1)\in \{0,1\}^{n-1}$ the mapping defined in \eqref{eq:f-alpha} satisfies $f_{\alpha_{2n}}=\ol{\rev}\bullet \pi$, the sequence of edges of the path $P$ when traversing it from its first to its last vertex has the form $(e_1,E_2,e_3,E_4,e_5)$, where $e_1$, $e_3$ and $e_5$ are single edges and $E_2$ and $E_4$ are sequences of edges that satisfy the following conditions:
There is a path $P'\in\cP_{2n}(n,n+1)$ and a path $P''\in\tcP_{2n}(n,n+1)$ (this is the crucial difference to before, where $P'$ and $P''$ where both from the set $\cP_{2n}(n,n+1)$) such that
\begin{align*}
  e_1 &= \big(S(P')\circ 00,S(P')\circ 01\big) \enspace, \\
  e_3 &= \big(L(P')\circ 01,L(P')\circ 11\big)=\big(L(P')\circ 01,\ol{\rev}(\pi(L(P'')))\circ 11\big) \enspace, \\
  e_5 &= \big(\ol{\rev}(\pi(F(P'')))\circ 11,\ol{\rev}(\pi(F(P'')))\circ 01\big) \enspace,
\end{align*}
$E_2$ is given by traversing the edges of $P'\circ 01$ starting at the vertex $S(P')\circ 01$ and ending at the vertex $L(P')\circ 01$, and $E_4$ is given by the traversing all the edges of $\ol{\rev}(\pi(P''))\circ 11$ in reverse order starting at the vertex $\ol{\rev}(\pi(L(P'')))\circ 11$ and ending at the vertex $\ol{\rev}(\pi(F(P'')))\circ 11$. 

As in the proof of Lemma~\ref{lemma:paths-original}, we distinguish five cases depending on which part of $P$ the edge $(u,v)$ belongs to. 
The last two cases are treated analogously to cases~3.4 and 3.5 in the proof of Lemma~\ref{lemma:paths-original}, using \eqref{eq:FtP} and \eqref{eq:LtP} instead of \eqref{eq:FP} and \eqref{eq:LP}.
The first three cases are exactly the same (not just analogous) as cases~3.1, 3.2 and 3.3. In particular, in case~3.2 where $(u,v)$ is an edge from $E_2$ note that the edges in $E_2$ originate from a path $P'\in\cP_{2n}(n,n+1)$, which is accounted for by setting the value of the variable $\flip$ to $\false$ in the recursive call in line~\ref{line:recurse-01} of the algorithm $\Paths()$.
\end{proof}

Given Lemma~\ref{lemma:paths-original} and Lemma~\ref{lemma:paths-flipped}, we may from now on use the notations $\cP_{2n}(k,k+1)$ and $\tcP_{2n}(k,k+1)$ interchangeably for the sets of paths defined in Sections~\ref{sec:recap-paths} and \ref{sec:recap-flippable} and for the sets of paths computed by the algorithm $\Paths()$ called with parameter $\flip=\false$ or $\flip=\true$, respectively.

\subsection{Correctness of the algorithm \texorpdfstring{$\HamCycle()$}{HamCycle()}}
\label{sec:correctness-hamcycle}

\subsubsection{$\HamCycle()$ computes a 2-factor}

By combining Lemma~\ref{lemma:paths-original}, property~(a) from Section~\ref{sec:recap-paths} and the relations \eqref{eq:FP} and \eqref{eq:LP}, we obtain that the paths $\cP_{2n}(n,n+1)$ computed by the algorithm $\Paths()$ indeed have properties~(i) and (ii) claimed in Section~\ref{sec:paths}. 
Consequently, assuming for a moment that the auxiliary functions $\IsFlipTree_1()$ and $\IsFlipTree_2()$ always return $\false$ (i.e., $\IsFlipVertex()$ always returns $\false$), the arguments given in Sections~\ref{sec:hamcyclenext} and \ref{sec:hamcycle} show that the algorithm $\HamCycle()$ correctly computes the 2-factor $\cC_{2n+1}$ defined in \eqref{eq:2-factor} in the middle levels graph $Q_{2n+1}(n,n+1)$.
We proceed to show that the algorithm $\HamCycle()$ computes a different 2-factor (but still a 2-factor) for each possible choice of boolean functions $\IsFlipTree_1()$ and $\IsFlipTree_2()$ that are called from within $\IsFlipVertex()$.
Later we prove that the functions $\IsFlipTree_1()$ and $\IsFlipTree_2()$ specified in Section~\ref{sec:isflipvertex} yield a 2-factor that consists only of a single cycle, i.e., a Hamilton cycle.
As already mentioned in Section~\ref{sec:hamcyclenext}, the different 2-factors are obtained by replacing some flippable pairs of paths from $\cP_{2n}(n,n+1)$ in the first set on the right hand side of \eqref{eq:2-factor} by the corresponding flipped pairs of paths from $\tcP_{2n}(n,n+1)$.
There are two potential problems that the function $\IsFlipVertex()$ as it is called from $\HamCycleNext()$ and $\HamCycle()$ (see line~\ref{line:check-flip} and line~\ref{line:init-flip-end}, respectively) could cause in this approach, and the next lemma shows that none of these problems occurs.
First, the function might return $\true$ for a vertex $x\in D_{2n}^{=0}(n)$ for which the path $P\in\cP_{2n}(n,n+1)$ that starts with this vertex (i.e., $F(P)=x$) is not contained in a flippable pair, so the subsequent calls to $\Paths()$ with parameter $\flip=\true$ would produce undefined output.
Second, given a flippable pair of paths $P,P'\in\cP_{2n}(n,n+1)$, the results of the calls $\IsFlipVertex(n,x)$ and $\IsFlipVertex(n,x')$ with $x:=F(P)$ and $x':=F(P')$ might be different/inconsistent, so our algorithm would not compute a valid 2-factor.

\begin{lemma}
\label{lemma:flip-possible}
Let $n\geq 1$, and let the sets $\cP_{2n}(n,n+1)$ and $\cX_{2n}(n,n+1)$ be as in Lemma~\ref{lemma:paths-flipped}. 
Furthermore, let $\IsFlipTree_1()$ and $\IsFlipTree_2()$ be arbitrary boolean functions on the sets of ordered rooted trees $\cT_{n,1}^*$ and $\cT_{n,2}^*$ defined in Section~\ref{sec:isflipvertex}, respectively, and let $\IsFlipVertex()$ be as defined in Algorithm~\ref{alg:isflipvertex}.
For any lattice path $x\in D_{2n}^{=0}(n)$ with $\IsFlipVertex(n,x)=\true$ there is another lattice path $x'\in D_{2n}^{=0}(n)$ with $\IsFlipVertex(n,x')=\true$ and two paths $P,P'\in\cP_{2n}(n,n+1)$ with $(F(P),F(P'))=(x,x')$ that form a flippable pair $(P,P')\in\cX_{2n}(n,n+1)$.
\end{lemma}

\begin{proof}
By the definitions in lines~\ref{line:h-inverse}--\ref{line:tau_2-image}, the tree $T:=h^{-1}(x)$ is contained in exactly one of the four sets $\cT_{n,1}^*$, $\tau_1(\cT_{n,1}^*)$, $\cT_{n,2}^*$, $\tau_2(\cT_{n,2}^*)$.
All these sets of trees are disjoint. 
We consider four cases depending on which of the sets the tree $T$ is contained in.

Case~1: If $T\in\cT_{n,1}^*$, then by the instructions in line~\ref{line:tau_1-preimage} we have $\IsFlipTree_1(T)=\true$. 
Now consider the tree $T':=\tau_1(T)$ and the lattice path $x':=h(T')=h(\tau_1(T))\in D_{2n}^{=0}(n)$. 
Clearly, the instructions in line~\ref{line:tau_1-image} ensure that $\IsFlipVertex(n,x')=\true$.
By the definition of $\cT_{n,1}^*$ and $\tau_1$, the lattice paths corresponding to the trees $T$ and $T'$ satisfy the preconditions of \cite[Lemma~21]{MR3483129}, implying that the pair of lattice paths $(h(T),h(T'))=(x,x')$ is contained in the set $H_1$ defined in \cite[Lemma~24]{MR3483129}, so by this lemma there are two paths $P,P'\in\cP_{2n}(n,n+1)$ with $(F(P),F(P'))=(x,x')$ that form a flippable pair $(P,P')\in\cX_{2n}(n,n+1)$.

Case~2: If $T\in\tau_1(\cT_{n,1}^*)$, then by the instructions in line~\ref{line:tau_1-image} we have $\IsFlipTree_1(\tau_1^{-1}(T))=\true$. 
Now consider the tree $T':=\tau_1^{-1}(T)$ and the lattice path $x':=h(T')=h(\tau_1^{-1}(T))\in D_{2n}^{=0}(n)$. 
Clearly, the instructions in line~\ref{line:tau_1-preimage} ensure that $\IsFlipVertex(n,x')=\true$.
From here the proof continues as in case~1, however, the roles of $T$ and $T'$ are interchanged.

The cases $T\in\cT_{n,2}^*$ and $T\in\tau_2(\cT_{n,2}^*)$ are dealt with analogously to case~1 and case~2, respectively, replacing the application of \cite[Lemma~21]{MR3483129} by \cite[Lemma~22]{MR3483129} and the set $H_1$ by $H_2$ defined in \cite[Lemma~24]{MR3483129}.
\end{proof}

The following simple but powerful lemma follows immediately from the definition of flippable/flipped pairs of paths given in Section~\ref{sec:paths}.

\begin{lemma}
\label{lemma:flip-paths}
Let $G$ be a graph, $\cC$ a 2-factor in $G$, $C$ and $C'$ two cycles in $\cC$, $(P,P')$ a flippable pair of paths such that $P$ is contained in $C$ and $P'$ is contained in $C'$, and $(R,R')$ a corresponding flipped pair of paths. 
Replacing the paths $P$ and $P'$ in the 2-factor $\cC$ by $R$ and $R'$ yields a 2-factor in which the two cycles $C$ and $C'$ are joined to a single cycle.
\end{lemma}

The Hamilton cycle computed by our algorithm $\HamCycle()$ is obtained from the 2-factor $\cC_{2n+1}$ defined in \eqref{eq:2-factor} by repeatedly applying the flipping operation from Lemma~\ref{lemma:flip-paths}. 
Specifically, we replace several flippable pairs of paths $P,P'\in\cP_{2n}(n,n+1)$ by the corresponding flipped paths $R,R'\in\tcP_{2n}(n,n+1)$ in the first set on the right hand side of \eqref{eq:2-factor} such that all cycles of the 2-factor $\cC_{2n+1}$ are joined to a single cycle.
This of course requires that all flippable pairs of paths are disjoint, which is guaranteed by the definitions of $\cP_{2n}(n,n+1)$ and $\cX_{2n}(n,n+1)$: The paths in $\cP_{2n}(n,n+1)$ are all disjoint, and each path from this set appears in at most one flippable pair in the set $\cX_{2n}(n,n+1)$.
To reasonably choose which flippable pairs of paths to replace by the corresponding flipped paths, we need to understand the cycle structure of the 2-factor $\cC_{2n+1}$.

\subsubsection{Structure of cycles in the 2-factor $\cC_{2n+1}$}

To describe the structure of the cycles in the 2-factor $\cC_{2n+1}$ defined in \eqref{eq:2-factor}, we introduce the concept of plane trees.

\textit{Plane trees.}
A \emph{plane tree} is a tree with a cyclic ordering of all neighbors of each vertex. 
We think of a plane tree as a tree embedded in the plane such that for each vertex $v$ the order in which its neighbors are encountered when walking around $v$ in counterclockwise direction is precisely the specified cyclic ordering, see the middle of Figure~\ref{fig:p-trees}.
We denote by $\cT_n$ the set of all plane trees with $n$ edges.
The notions of thin/thick leaves and of clockwise/counterclockwise-next leaves introduced in Section~\ref{sec:isflipvertex} for ordered rooted trees can be defined for plane trees in a completely analogous fashion.

\textit{Transformations $\plane()$ and $\troot()$ between ordered rooted trees and plane trees.}
For any ordered rooted tree $T^*\in\cT_n^*$, we define a plane tree $T=\plane(T^*)\in\cT_n$ as follows: The underlying (abstract) tree of $T^*$ and $T$ is the same. 
For the root $r$ of $T^*$, the cyclic ordering of neighbors of $r$ in $T$ is given by the left-to-right ordering of the children of $r$ in $T^*$. 
For any other vertex $v$ of $T^*$, if $(u_1,u_2,\ldots,u_k)$ is the left-to-right ordering of the children of $v$ in $T^*$ and $w$ is the parent of $v$, then we define $(w,u_1,u_2,\ldots,u_k)$ as the cyclic ordering of neighbors of $v$ in $T$.

For any plane tree $T\in\cT_n$ and any edge $(r,u)$ of $T$, we define an ordered rooted tree $T^*=\troot(T,(r,u))\in\cT_n^*$ as follows: The underlying (abstract) tree of $T$ and $T^*$ is the same.
The vertex $r$ is the root of $T^*$, and if $(u_1,u_2,\ldots,u_k)$ with $u_1=u$ is the cyclic ordering of neighbors of $r$ in $T$, then $(u_1,u_2,\ldots,u_k)$ is the left-to-right ordering of the children of $r$ in $T^*$. 
For any other vertex $v$ of $T$, if $w$ is the vertex on the path from $v$ to $r$ and $(u_0,u_1,\ldots,u_k)$ with $u_0=w$ is the cyclic ordering of neighbors of $v$ in $T$, then $(u_1,u_2,\ldots,u_k)$ is the left-to-right ordering of the children of $v$ in $T^*$.

Informally speaking, $\plane(T^*)$ is obtained from $T^*$ by `forgetting' the root vertex, and $\troot(T,(r,u))$ is obtained from $T$ by `pulling out' the vertex $r$ as root such that $u$ is the leftmost child of the root, see Figure~\ref{fig:p-trees}.

\begin{figure}
\centering
\PSforPDF{
 \psfrag{t1s}{$T_1^*$}
 \psfrag{t2s}{$T_2^*$}
 \psfrag{tp}{$T$}
 \psfrag{u1}{$u_1$}
 \psfrag{u2}{$u_2$}
 \psfrag{r1}{$r_1$}
 \psfrag{r2}{$r_2$}
 \psfrag{plane1}{$\plane(T_1^*)$}
 \psfrag{plane2}{$\plane(T_2^*)$}
 \psfrag{root1}{$\troot(T,(r_1,u_1))$}
 \psfrag{root2}{$\troot(T,(r_2,u_2))$}
 \includegraphics{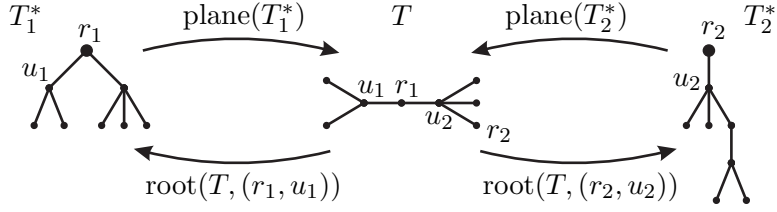}
}
\caption{Two ordered rooted trees $T_1^*,T_2^*\in\cT_7^*$ (left and right), a plane tree $T\in\cT_7$ (middle), and the transformations between them.}
\label{fig:p-trees}
\end{figure}

For any cycle $C$ of the 2-factor $\cC_{2n+1}$ defined in \eqref{eq:2-factor} we let $F(C)$ denote the cyclic sequence of all vertices of the form $F(P)$, $P\in\cP_{2n}(n,n+1)$, for which the path $P\circ 0$ is contained in $C$ (when walking along $C$ in the direction such that $P\circ 0$ is traversed starting at the first vertex and ending at the last vertex).
It turns out to be very useful to consider the ordered rooted trees corresponding to the elements of the sequence $F(C)$ (recall \eqref{eq:FP} and the bijection between lattice paths in $D_{2n}^{=0}(n)$ and ordered rooted trees in $\cT_n^*$).
The following lemma is a slight reformulation of \cite[Lemma~20]{MR3483129}.
For the reader's convenience, the lemma is illustrated in Figure~\ref{fig:rot-trees}.

\begin{lemma}
\label{lemma:2f-C1}
For any $n\geq 1$ and any cycle $C$ of the 2-factor $\cC_{2n+1}$ defined in \eqref{eq:2-factor}, considering the cyclic sequence of ordered rooted trees $(T_1,T_2,\ldots,T_k):=F(C)$ we have $h^{-1}(T_{i+1})=\trot(h^{-1}(T_i))$ for all $i=1,2,\ldots,k$ (indices are considered modulo $k$) with $h^{-1}$ and $\trot$ as defined in Section~\ref{sec:isflipvertex}.
I.e., we can associate the cycle $C$ with the plane tree $\plane(h^{-1}(T_1))=\cdots=\plane(h^{-1}(T_k))\in\cT_n$. 
Moreover, for any plane tree $T\in\cT_n$ there is a cycle $C\in\cC_{2n+1}$ to which the tree $T$ is associated. 
In particular, the number of cycles of the 2-factor is $|\cC_{2n+1}|=|\cT_n|$.
\end{lemma}

\begin{figure}
\centering
\PSforPDF{
 \psfrag{2f1}{$\cC_{2n+1}=\{C_1,C_2,C_3\}$}
 \psfrag{ch1}{$C_1$}
 \psfrag{ch2}{$C_2$}
 \psfrag{ch3}{$C_3$} 
 %\psfrag{tch1}{$T(C_1)$}
 %\psfrag{tch2}{$T(C_2)$}
 %\psfrag{tch3}{$T(C_3)$}
 \psfrag{rot}{$\trot$}
 \psfrag{h}{\Large $h$}
 \psfrag{t01}{$\That_1$}
 \psfrag{t02}{$\That_2$}
 \psfrag{t03}{$\That_3$}
 \psfrag{t04}{$\That_4$}
 \psfrag{t05}{$\That_5$}
 \psfrag{t06}{$\That_6$}
 \psfrag{t07}{$\That_7$}
 \psfrag{t08}{$\That_8$}
 \psfrag{t09}{$\That_9$}
 \psfrag{t10}{$\That_{10}$}
 \psfrag{tel}{$\That_{11}$}
 \psfrag{t12}{$\That_{12}$}
 \psfrag{t13}{$\That_{13}$}
 \psfrag{t14}{$\That_{14}$}
 \psfrag{th01}{$T_1$}
 \psfrag{th02}{$T_2$}
 \psfrag{th03}{$T_3$}
 \psfrag{th04}{$T_4$}
 \psfrag{th05}{$T_5$}
 \psfrag{th06}{$T_6$}
 \psfrag{th07}{$T_7$}
 \psfrag{th08}{$T_8$}
 \psfrag{th09}{$T_9$}
 \psfrag{th10}{$T_{10}$}
 \psfrag{thel}{$T_{11}$}
 \psfrag{th12}{$T_{12}$}
 \psfrag{th13}{$T_{13}$}
 \psfrag{th14}{$T_{14}$}
 \psfrag{fp01}{\footnotesize $F(P_1)$}
 \psfrag{fp02}{\footnotesize $F(P_2)$}
 \psfrag{fp03}{\footnotesize $F(P_3)$}
 \psfrag{fp04}{\footnotesize $F(P_4)$}
 \psfrag{fp05}{\footnotesize $F(P_5)$}
 \psfrag{fp06}{\footnotesize $F(P_6)$}
 \psfrag{fp07}{\footnotesize $F(P_7)$}
 \psfrag{fp08}{\footnotesize $F(P_8)$}
 \psfrag{fp09}{\footnotesize $F(P_9)$}
 \psfrag{fp10}{\footnotesize $F(P_{10})$}
 \psfrag{fpel}{\footnotesize $F(P_{11})$}
 \psfrag{fp12}{\footnotesize $F(P_{12})$}
 \psfrag{fp13}{\footnotesize $F(P_{13})$}
 \psfrag{fp14}{\footnotesize $F(P_{14})$}
 \psfrag{fph01}{\footnotesize $F(P_1)$}
 \psfrag{fph02}{\footnotesize $F(P_2)$}
 \psfrag{fph03}{\footnotesize $F(P_3)$}
 \psfrag{fph04}{\footnotesize $F(P_4)$}
 \psfrag{fph05}{\footnotesize $F(P_5)$}
 \psfrag{fph06}{\footnotesize $F(P_6)$}
 \psfrag{fph07}{\footnotesize $F(P_7)$}
 \psfrag{fph08}{\footnotesize $F(P_8)$}
 \psfrag{fph09}{\footnotesize $F(P_9)$}
 \psfrag{fph10}{\footnotesize $F(P_{10})$}
 \psfrag{fphel}{\footnotesize $F(P_{11})$}
 \psfrag{fph12}{\footnotesize $F(P_{12})$}
 \psfrag{fph13}{\footnotesize $F(P_{13})$}
 \psfrag{fph14}{\footnotesize $F(P_{14})$}
 \psfrag{pu}[c][c][1][90]{$\cP_{2n}(n,n+1)$}
 \psfrag{pl}[c][c][1][90]{$\ol{\rev}(\cP_{2n}(n,n+1))$}
 \includegraphics{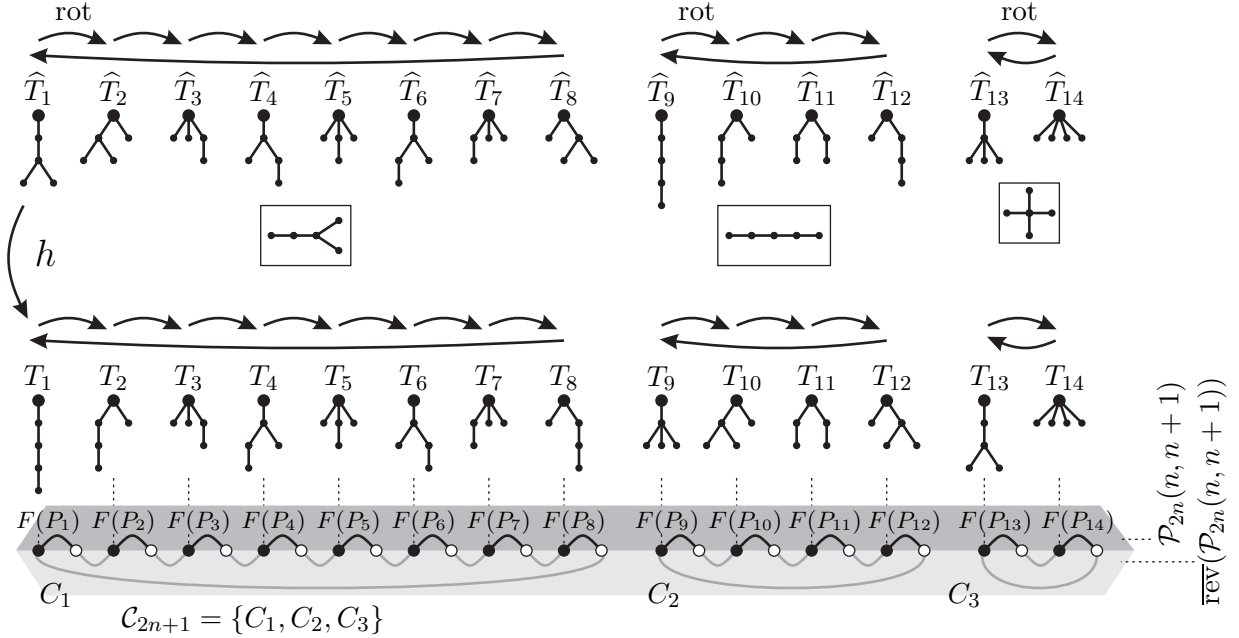}
}
\caption{Structure of the 2-factor $\cC_{2n+1}$ for $n=4$ (cf.\ the bottom part of Figure~\ref{fig:2factor}). 
The figure shows the cyclic sequences $F(C)$ of ordered rooted trees for all cycles $C\in\cC_{2n+1}$ (bottom) and their preimages under the mapping $h$ (top). 
Note that trees rotate along each of the sequences $h^{-1}(F(C))$, $C\in\cC_{2n+1}$.
The boxes show the plane trees associated with each cycle $C\in\cC_{2n+1}$.
The sets of trees $\cT_{n,1}^*$ and $\cT_{n,2}^*$ defined in Section~\ref{sec:isflipvertex} are given by $\cT_{n,1}^*=\{\That_7,\That_{11}\}$ and $\cT_{n,2}^*=\{\That_6\}$ and we have $\tau_1(\That_7)=\That_{14}$, $\tau_1(\That_{11})=\That_3$ and $\tau_2(\That_6)=\That_9$.
}
\label{fig:rot-trees}
\end{figure}

\subsubsection{$\HamCycle()$ computes a Hamilton cycle}

Lemmas~\ref{lemma:flip-possible}, \ref{lemma:flip-paths} and \ref{lemma:2f-C1} motivate the following definitions of two directed multigraphs $\cG_n$ and $\cH_n\seq\cG_n$ that are illustrated in Figure~\ref{fig:gh6}.

\textit{Graphs $\cG_n$ and $\cH_n$.}
The nodes of $\cG_n$ are sets of ordered rooted trees from $\cT_n^*$ such that all trees that differ by the rotation operation $\trot$ are in the same set.
The edges of $\cG_n$ are defined as follows: For $i\in\{1,2\}$ and any pair of trees $(T,T')$ with $T\in\cT_{n,i}^*$ and $T'=\tau_i(T)$ we add an edge directed from the set of trees containing $T$ to the set of trees containing $T'$. 
We refer to the edges of $\cG_n$ induced by the mappings $\tau_1$ and $\tau_2$ as $\tau_1$-edges or $\tau_2$-edges, respectively.
The graph $\cH_n$ has the same nodes as $\cG_n$, but only a subset of its edges. 
Specifically, for $i\in\{1,2\}$ and any pair of trees $(T,T')$ with $T\in\cT_{n,i}^*$ and $T'=\tau_i(T)$ we only add an edge directed from the set of trees containing $T$ to the set of trees containing $T'$ if $\IsFlipTree_i(T)=\true$, where $\IsFlipTree_1()$ and $\IsFlipTree_2()$ are as defined in Section~\ref{sec:isflipvertex}.

As each node of $\cG_n$ and $\cH_n$ is a set $\{T_1,T_2,\ldots,T_k\}$ of ordered rooted trees that differ by rotation, i.e., we have $T_{i+1}=\trot(T_i)$ for $i=1,2,\ldots,k$, we can identify this node with the plane tree $\plane(T_1)=\cdots=\plane(T_k)\in\cT_n$, see the example in Figure~\ref{fig:gh6} for $n=6$.
Also, we can interpret the result of the operations $\tau_1$ and $\tau_2$ that define the edges of $\cG_n$ and $\cH_n$ in terms of the corresponding plane trees.
Note e.g.\ that $\tau_1$ increases the number of leaves by one, and $\tau_2$ decreases the number of leaves by one, implying that neither $\cG_n$ nor $\cH_n$ has any loops.
This effect can be seen in Figure~\ref{fig:gh6}, where all $\tau_1$-edges go from right to left, and all $\tau_2$-edges from left to right.
The graph $\cG_n$ may have multiple edges between nodes, see Figure~\ref{fig:gh6}.
The $\tau_1$- and $\tau_2$-edges added to the graph $G_n$ are also described in Figure~\ref{fig:rot-trees} for the smaller example $n=4$.

By Lemma~\ref{lemma:2f-C1}, the nodes of $\cG_n$ and $\cH_n$ correspond to the cycles of the 2-factor $\cC_{2n+1}$ defined in \eqref{eq:2-factor}. 
Applying Lemma~\ref{lemma:flip-possible} with boolean functions $\IsFlipTree_1()$ and $\IsFlipTree_2()$ that always return $\true$ shows that each edge of $\cG_n$ corresponds to a flippable pair of paths from the set $\cX_{2n}(n,n+1)$ for which the paths are contained in the corresponding cycles.
In other words, the edges of $\cG_n$ capture all potential flipping operations that could be performed to modify the 2-factor $\cC_{2n+1}$.
Applying Lemma~\ref{lemma:flip-possible} with the boolean functions $\IsFlipTree_1()$ and $\IsFlipTree_2()$ defined in Section~\ref{sec:isflipvertex} shows that each edge of $\cH_n$ corresponds to a flippable pair of paths from $\cP_{2n}(n,n+1)$ that are actually replaced by the corresponding flipped paths from $\tcP_{2n}(n,n+1)$ in the first set on the right hand side of \eqref{eq:2-factor} by our Hamilton cycle algorithm $\HamCycle()$. 
By Lemma~\ref{lemma:flip-paths} the edges of $\cH_n$ therefore indicate which pairs of cycles from the 2-factor $\cC_{2n+1}$ the algorithm joins to a single cycle.
Consequently, to complete the correctness proof for the algorithm $\HamCycle()$, it suffices to show the following lemma.

\begin{figure}
\centering
\PSforPDF{
 \psfrag{g6}{\Large $\cG_6$}
 \psfrag{h6}{\Large $\cH_6$}
 \psfrag{t1}{$\tau_1$-edges}
 \psfrag{t2}{$\tau_2$-edges}
 \psfrag{pn}{$P_n$}
 \psfrag{pnp}{$P_n'$}
 \psfrag{pnpp}{$P_n''$}
 \psfrag{6l}{6 leaves}
 \psfrag{5l}{5 leaves}
 \psfrag{4l}{4 leaves}
 \psfrag{3l}{3 leaves}
 \psfrag{2l}{2 leaves}
 \includegraphics{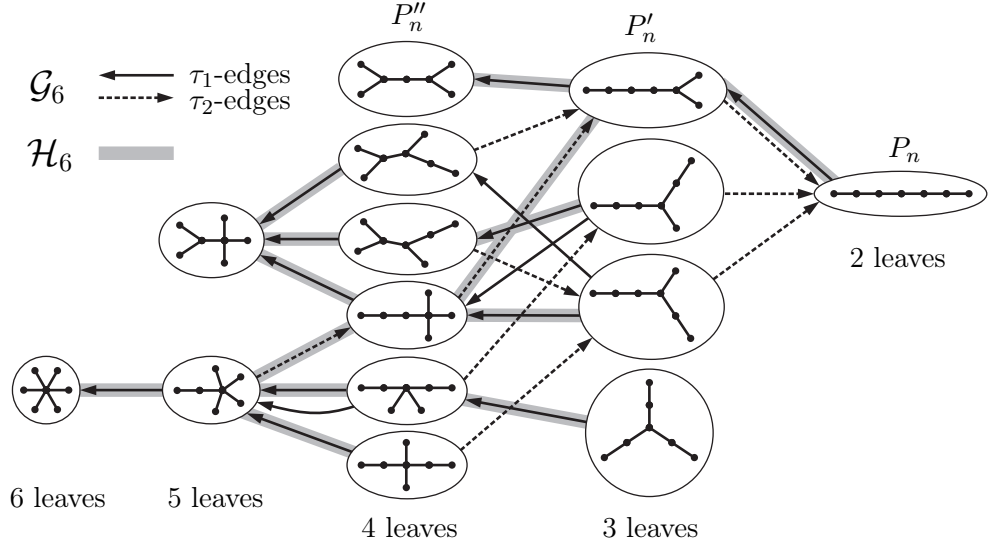}
}
\caption{The graphs $\cG_6$ and $\cH_6$ with nodes are arranged in levels according to the number of leaves of the corresponding plane trees. 
The figure shows for each node, i.e., for each set of ordered rooted trees that differ by rotation, the plane tree corresponding to these trees. 
In the figure, $\tau_1$-edges are drawn as solid lines, $\tau_2$-edges as dashed lines, and the edges of $\cH_6$ are highlighted in grey.}
\label{fig:gh6}
\end{figure}

\begin{lemma}
\label{lemma:spanning-tree}
For any $n\geq 1$, the graph $\cH_n$ is a spanning tree of $\cG_n$.
\end{lemma}

\begin{proof}
For $n=1$ and $n=2$ the graph $\cG_n$ consists only of a single node and no edges, and we have $\cG_n=\cH_n$, so in these cases the claim is trivially true. 
For $n=3$ the graph $\cG_n$ consists of two nodes, connected by a $\tau_1$-edge, and we again have $\cG_n=\cH_n$, proving the lemma also in this case. 
For the rest of the proof we assume that $n\geq 4$.

Let $\cH_n'$ denote the graph obtained from $\cH_n$ by removing all $\tau_2$-edges.
We first show that $\cH_n'$ is a spanning forest of $\cG_n$.

By the definition of the function $\IsFlipTree_1()$ from Section~\ref{sec:isflipvertex}, the outdegree of every node $T\in\cT_n$ in $\cH_n'$ is at most 1. 
More precisely, if $T$ has a thin leaf, then the outdegree is 1, as there is a way to root $T$ such that it is contained in the set $\cT_{n,1}^*$ (let $u$ denote a thin leaf, $u'$ the neighbor of $u$ and $v$ the neighbor of $u'$, then this rooted tree is given by $\troot(T,(v,u'))$, see the left hand side of Figure~\ref{fig:tau12}), and only for the lexicographically smallest such tree the function returns $\true$. 
On the other hand, it $T$ has no thin leaves, then the outdegree is 0, as there is no way to root $T$ such that it is contained in $\cT_{n,1}^*$, so the function returns $\false$ for all these trees (see Figure~\ref{fig:gh6}).
Moreover, as $\tau_1$ increases the number of leaves by one, the graph $\cH_n'$ is acyclic. 
An acyclic directed graph with outdegree at most 1 at every node is a forest, proving that $\cH_n'$ is indeed a spanning forest of $\cG_n$.

In the following we argue that adding the $\tau_2$-edges from $\cH_n$ to $\cH_n'$ turns the spanning forest into a spanning tree of $\cG_n$.

By what we said before the components (=trees) of the forest $\cH_n'$ can be characterized precisely: Every component has a unique sink, i.e., a node with no outgoing edges, which corresponds to a plane tree with no thin leaves. 
Also, the distance of a node $T\in\cT_n$ from the sink of the component it belongs to equals the number of thin leaves of $T$.
Moreover, we obtain that the number of components of $\cH_n'$ is given by the number of plane trees from $\cT_n$ that have no thin leaves. 
In Figure~\ref{fig:gh6}, $\cH_n'$ is obtained from $\cH_n$ by removing two dashed $\tau_2$-edges, resulting in three components whose sinks appear as the leftmost node of each component.

We now consider a component of $\cH_n'$ that will play a special role in the following.
Let $P_n\in\cT_n$ denote the path with $n$ edges, and let $(v_1,v_2,\ldots,v_{n+1})$ be the sequence of vertices along this path.
We denote by $P_n'\in\cT_n$ the graph obtained from $P_n$ by replacing the edge $(v_n,v_{n+1})$ by the edge $(v_{n-1},v_{n+1})$. 
Furthermore, we denote by $P_n''\in\cT_n$ the graph obtained from $P_n'$ by replacing the edge $(v_1,v_2)$ by the edge $(v_1,v_3)$.
Clearly, $\cH_n'$ has a component consisting of the three nodes $P_n$, $P_n'$ and $P_n''$, with a $\tau_1$-edge directed from $P_n$ to $P_n'$ and a $\tau_1$-edge directed from $P_n'$ to the sink $P_n''$, see the top part of Figure~\ref{fig:gh6}.
We denote this component of $\cH_n'$ by $\cS_n$. 
Note that the definition of $\cS_n$ is meaningful only for $n\geq 4$.

We proceed to show the following:
(i) In every component of $\cH_n'$ except $\cS_n$ there is exactly one node that has an outgoing $\tau_2$-edge in $\cH_n$, and the component $\cS_n$ has no node with an outgoing $\tau_2$-edge.
(ii) The $\tau_2$-edges of $\cH_n$ form an acyclic graph on the components of $\cH_n'$ (in particular, every $\tau_2$-edge of $\cH_n$ starts and ends in a different component of $\cH_n'$).
Combining our earlier observation that $\cH_n'$ is a spanning forest of $\cG_n$ with claims~(i) and (ii) clearly proves the lemma, so it remains to prove (i) and (ii).

Proof of (i):
Consider a component of $\cH_n'$ with sink $T\in\cT_n$, and let $N(T)$ be the set of all nodes in distance 1 from $T$ in $\cH_n'$. 
As $T$ has no thin leaves, it can be rooted so that it is an image of $\tau_1$: let $u$ and $u'$ be two leaves that have a common neighbor $v$ such that $u'$ is the counterclockwise-next leaf from $u$, then this rooted tree is given by $\troot(T,(v,u))$. 
See the left hand side of Figure~\ref{fig:tau12}.
The preimage is a tree with exactly one thin leaf, which corresponds to a node $T'\in\cT_n$ with a $\tau_1$-edge directed from $T'$ to $T$ that is present in the graph $\cH_n'$. 
It follows that the set $N(T)$ is nonempty. 
Note that the function $\IsFlipTree_2()$ returns $\true$ only if the given tree has a single thin leaf, so the only outgoing $\tau_2$-edges of $\cH_n$ that start from a node in the same component as $T$ start at a node in $N(T)$. 
Moreover, as any tree in $N(T)$ can be rooted uniquely so that it is a preimage of $\tau_2$ (let $v$ denote the thin leaf, $u$ the counterclockwise-next leaf from $v$, $u'$ the neighbor of $u$ and $r$ the neighbor of $u'$ next to $u$ in the counterclockwise ordering of the neighbors of $u'$, then this rooted tree is given by $\troot(T,(r,u'))$, see the right hand side of Figure~\ref{fig:tau12}), every node in $N(T)$ has exactly one outgoing $\tau_2$-edge in $\cG_n$, see Figure~\ref{fig:gh6}.
If the component containing $T$ is different from $\cS_n$, then the function $\IsFlipTree_2()$ returns $\true$ for exactly one of these rooted trees: The function first computes a rooted version of $T$ for each of them, thus consistently assigning an integer value to each of the trees (the integers are derived from distances between certain vertices of $T$, but the origin of these values is irrelevant for this part of the argument), and returns $\true$ only for the tree that was assigned the largest value and that is lexicographically smallest among all trees with the same value.
However, if the component containing $T$ is $\cS_n$ (i.e., $T=P_n$), then $N(T)=\{P_n'\}$ and the unique way to root $P_n'$ such that it is a preimage of $\tau_2$ is the tree $1^{n-1}\circ 0^{n-2}\circ 100$, for which the function $\IsFlipTree_2()$ exceptionally returns $\false$.
This exceptional rooted tree is not encountered in any component other than $\cS_n$.

\begin{figure}
\centering
\PSforPDF{
 \psfrag{gn}{\Large $\cG_n$}
 \psfrag{hn}{\Large $\cH_n$}
 \psfrag{t1}{$\tau_1$-edges}
 \psfrag{t2}{$\tau_2$-edges} 
 \psfrag{n2}{2}
 \psfrag{n5}{5}
 \psfrag{n6}{6}
 \psfrag{n7}{7}
 \psfrag{n8}{8}
 \psfrag{u}{$u$}
 \psfrag{up}{$u'$}
 \psfrag{v}{$v$}
 \psfrag{vp}{$v'$}
 \psfrag{vpp}{$v''$}
 \psfrag{w}{$w$}
 \psfrag{t}{$T$}
 \psfrag{dt}{$d(T)=\max\{2,5,6,7\}=7$}
 \psfrag{tp}{$T'$}
 \psfrag{dtp}{$d(T')\geq 8>d(T)$}
 \psfrag{th}{$\That$}
 \psfrag{thp}{$\That'$}
 \includegraphics{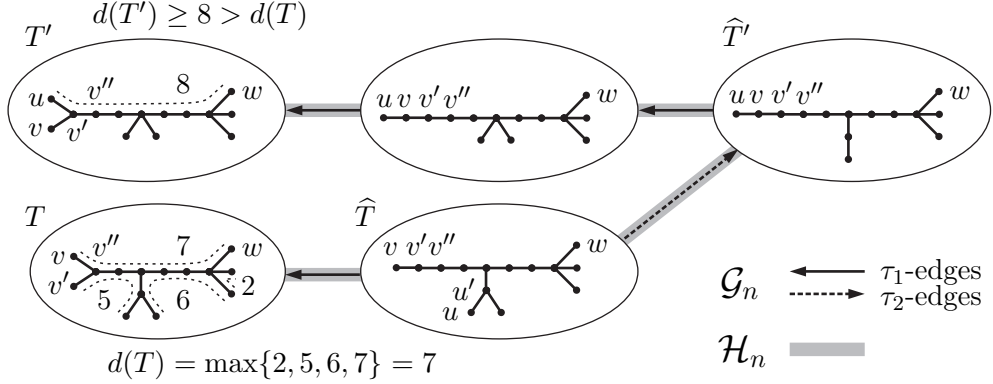}
}
\caption{Illustration of the notations used in the proof of Lemma~\ref{lemma:spanning-tree}.}
\label{fig:dT}
\end{figure}

Proof of (ii):
For the reader's convenience, this part of the proof is illustrated in Figure~\ref{fig:dT}.
For any plane tree $T\in\cT_n$, we define $d(T)$ as the maximum distance between any two leaves $v$ and $w$ that satisfy the following conditions: $w$ is the clockwise-next leaf from $v$ and the counterclockwise-next leaf $v'$ from $v$ has distance 2 from $v$.
In other words, the leaves $v$ and $v'$ are adjacent to a common vertex $v''$; see the bottom left part of Figure~\ref{fig:dT}.
Claim~(ii) is an immediate consequence of the following observation: For any $\tau_2$-edge from $\cH_n$ directed from a component of $\cH_n'$ with sink $T\in\cT_n$ to a component of $\cH_n'$ with sink $T'\in\cT_n$, we have $d(T)<d(T')$.
In particular, the edge starts and ends in a different component of $\cH_n'$.
To see this let $\That\in\cT_n$ be the node in distance 1 of $T$ in $\cH_n'$ that is the starting node of a $\tau_2$-edge of $\cH_n$ (this node exists and is unique by claim~(i)), and let $\That'\in\cT_n$ be the end node of this $\tau_2$-edge. 
As $\That$ has exactly one thin leaf, and as $\tau_2$ preserves one thin leaf and possibly creates one additional thin leaf, $\That'$ has distance 1 or 2 from $T'$. 
So in $\cH_n$ there is a $\tau_1$-edge directed from $\That$ to $T$, a $\tau_2$-edge directed from $\That$ to $\That'$ and either a single $\tau_1$-edge or a path consisting of two $\tau_1$-edges directed from $\That'$ to $T'$, see Figure~\ref{fig:dT}.

By the definition of $\IsFlipTree_2()$ and of the parameter $d(T)$, the tree $T$ has two leaves $v$ and $w$ in distance $d(T)$ such that $w$ is the clockwise-next leaf from $v$ and the counterclockwise-next leaf $v'$ from $v$ has distance 2 from $v$, implying that $v$ and $v'$ have a common adjacent vertex $v''$.
Furthermore $\That$ is obtained from $T$ by replacing the edge $(v,v'')$ by $(v,v')$.
The function $\IsFlipTree_2()$ of course computes a rooted version of $T$ from a rooted version of $\That$. 

By the definition of $\tau_2$, the tree $\That'$ is obtained from $\That$ by considering the counterclockwise-next leaf $u$ from $v$, its neighbor $u'$ in $\That$, and by replacing the edge $(u,u')$ by the edge $(u,v)$. 
By the definition of $\tau_1$, $T'$ is obtained from $\That'$ by replacing the edge $(u,v)$ by $(u,v')$ plus possibly another application of $\tau_1$ if $\That'$ has two thin leaves. 
It follows that the distance between $u$ and $w$ in $T'$ is $d(T)+1$, implying that $d(T')>d(T)$, as claimed.

This completes the proof.
\end{proof}

\begin{remark}
\label{remark:other-cycles}
Here is how our algorithm can be modified to compute another Hamilton cycle in the middle levels graph: The key is to modify the functions $\IsFlipTree_1()$ and $\IsFlipTree_2()$ such that they define a spanning tree of $\cG_n$ that is different from the `canonical' spanning tree $\cH_n$. 
Each different spanning tree of $\cG_n$ clearly yields a different Hamilton cycle in the middle levels graph $Q_{2n+1}(n,n+1)$ (cf.~\cite[Lemma~6]{MR3483129}).
The only limitation an efficient algorithm must obey is that the spanning tree must be locally computable, i.e., given information about the current node $T\in\cT_n$ of $\cG_n$, the algorithm must decide based on a polynomially sized neighborhood of $T$ which edges of $\cG_n$ to include and which to exclude from a spanning tree.
Note that the degree of the nodes of $\cG_n$ is $\cO(n)$ and the total number of nodes is $|\cT_n|=\Theta(4^n\cdot n^{-5/2})$.
\end{remark}

\section{Running time and space requirements of the algorithm}
\label{sec:running-time}

In this section we prove that our Hamilton cycle algorithm can be implemented so that it satisfies the runtime and space bounds claimed in the introduction. 
In the runtime analysis we apply some results about the paths $\cP_{2n}(n,n+1)$ that have been established in \cite{muetze-weber:12}.

\subsection{Running time}

A naive implementation of the function $\Paths()$ takes time $\cO(n^2)$: To see this observe that the membership tests whether $x^-$ is contained in one of the sets $D_{2n-2}^{=0}(n-1)$, $D_{2n-2}^-(n-1)$ or $D_{2n-2}^{>0}(n)$ in lines~\ref{line:xp00f}, \ref{line:xp01f1}, \ref{line:xp01f2}, \ref{line:xp01f3}, \ref{line:xp11f1} and \ref{line:xp11f2} and the application of the mappings $\ol{\rev}$ and $\pi$ in line~\ref{line:recurse-11} take time $\cO(n)$ (recall that $\ol{\rev}^{-1}=\ol{\rev}$ and $\pi^{-1}=\pi$), and that the value of $n$ decreases by 1 with each recursive call.
In the following we sketch how this can improved so that each call of $\Paths()$ takes only time $\cO(n)$.
More details can be found in the comments of our C++ implementation~\cite{www}.
For this we maintain counters $c_0$ and $c_1$ for the number of zeros and ones of a given bitstring $x=(x_1,x_2,\ldots,x_{2n})$. 
Moreover, interpreting the bitstring $x$ as a lattice path (as described in Section~\ref{sec:basic-defs}), we maintain vectors $c_{00}$, $c_{01}$, $c_{10}$, $c_{11}$ that count the number of occurences of pairs of consecutive bits $(x_{2i},x_{2i+1})$, $i\in\{1,2,\ldots,n-1\}$, \emph{per height level of the lattice path} for each of the four possible value combinations of $x_{2i}$ and $x_{2i+1}$. 
E.g., for the bitstring $x=1100001010$ the vector $c_{10}$ has a single nonzero entry 1 at height level (=index) 1 for the two bits $(x_2,x_3)=10$, the vector $c_{00}$ has a single nonzero entry 1 at height level (=index) 0 for the two bits $(x_4,x_5)=00$, and the vector $c_{01}$ has a single nonzero entry 2 at height level (=index) $-1$ for the pairs of bits $(x_6,x_7)=(x_8,x_9)=01$.
Using these counters, the three membership tests mentioned before can be performed in constant time.
E.g., a bitstring $x=(x_1,x_2,\ldots,x_{2n})$ is contained in $D_{2n}^{=0}(n)$ if and only if $c_0=c_1$ and $x_1=1$ and the entry of $c_{00}$ at height level 0 equals 0 (i.e., the lattice path never moves below the line $y=0$).
Moreover, these counters can be updated in constant time when removing the last two bits of $x$ and when applying the mappings $\ol{\rev}$ and $\pi$: Note that $\ol{\rev}$ simply swaps the roles of $c_0$ and $c_1$ and the roles of $c_{00}$ and $c_{11}$ and possibly incurs an index shift, and that $\pi$ simply swaps the roles of $c_{10}$ and $c_{01}$.
To compute the applications of $\ol{\rev}$ and $\pi$ in line~\ref{line:recurse-11} in constant time, we do not modify $x$ at all, but rather count the number of applications of $\ol{\rev}$ and $\pi$ and keep track of the middle range of bits of $x$ that is still valid.
When removing the last two bits of $x$, this range shrinks by 2 on one of the sides. 
Taking into account that multiple applications of $\ol{\rev}$ and $\pi$ cancel each other out, this allows us to compute the effect of applying those mappings \emph{lazily} when certain bits are queried later on (when testing the values of the last two bits of some substring of $x$).
E.g., by repeatedly removing the last two bits and applying the mappings $\ol{\rev}$ and $\pi$, the bitstring $x=(x_1,x_2,\ldots,x_{10})$ is transformed into $(\ol{x_8},\ol{x_6},\ol{x_7},\ol{x_4},\ol{x_5},\ol{x_2},\ol{x_3},\ol{x_1})$, $(x_2,x_4,x_5,x_6,x_7,x_8)$, $(\ol{x_6},\ol{x_4},\ol{x_5},\ol{x_2})$, and finally into $(x_4,x_6)$.

The function $\IsFlipVertex()$ can be implemented to run in time $\cO(n^2)$.
To see this observe that the result of applying the function $h^{-1}$ in line~\ref{line:h-inverse} can be computed in time $\cO(n^2)$, and that the functions $\IsFlipTree_1()$ and $\IsFlipTree_2()$ called in lines~\ref{line:tau_1-preimage}--\ref{line:tau_2-image} also need time $\cO(n^2)$: For both we need to rotate an ordered rooted tree with $n$ edges (=bitstring of length $2n$) for one full rotation, and each rotation operation takes time $\cO(n)$.

It was shown in \cite[Lemma~9]{muetze-weber:12} that the length of any path $P\in\cP_{2n}(n,n+1)$ with a first vertex $F(P)=:x\in D_{2n}^{=0}(n)$ is given by the following simple formula: Considering the unique decomposition $x=1\circ x_\ell\circ 0\circ x_r$ with $x_\ell\in D_{2k}^{=0}(k)$ for some $k\geq 0$, the length of $P$ is given by $2|x_\ell|+2\leq 2(2n-2)+2=4n-2$.
It follows that the while-loop in line~\ref{line:while-not-start} terminates after at most $\cO(n)$ iterations, i.e., the initialization phase of $\HamCycle()$ (lines~\ref{line:init-flip-start}--\ref{line:init-flip-end}) takes time $\cO(n^2)$.

It was shown in \cite[Theorem~10]{muetze-weber:12} that the distance between any two neighboring vertices of the form $x\circ 0$, $x'\circ 0$ with $x,x'\in D_{2n}^{=0}(n)$ ($x$ and $x'$ are first vertices of two paths $P,P'\in\cP_{2n}(n,n+1)$) on a cycle in \eqref{eq:2-factor} is exactly $4n+2$.
Comparing the lengths of two paths from $\cP_{2n}(n,n+1)$ that form a flippable pair with the lengths of the corresponding flipped paths from the set $\tcP_{2n}(n,n+1)$, we observe that either the length of one the paths decreases by $4$ and the length of the other increases by $4$, or the lengths of the paths do not change.
Specifically, the paths $P,P'$ and $R,R'$ defined in \eqref{eq:P423} and \eqref{eq:tP423} have exactly the length differences $-4$ and $+4$, and these differences only propagate through the first cases of the $\Paths()$ recursion, but not the last case in lines~\ref{line:xp11}--\ref{line:recurse-11}.
It follows that every call of $\HamCycleNext()$ for which the condition in line~\ref{line:check-flip} is satisfied and which therefore takes time $\cO(n^2)$ due to the call of $\IsFlipVertex()$, is followed by at least $4n-3$ calls in which the condition is not satisfied, in which case $\HamCycleNext()$ terminates in time $\cO(n)$. 
Consequently, $\ell$ consecutive calls of $\HamCycleNext()$ take time $\cO(n^2+n\ell)$.

Summing up the time $\cO(n^2)$ spent for the initialization phase and $\cO(n^2+n\ell)$ for the actual Hamilton cycle computation, we conclude that the algorithm $\HamCycle(n,x,\ell)$ runs in time $\cO(n^2+n\ell)=\cO(n\ell(1+\frac{n}{\ell}))$, as claimed.

\subsection{Space requirements}

The optimized variant of the algorithm $\Paths()$ discussed in the previous section requires space $\cO(n)$, e.g., to store the counting vectors $c_{00}$, $c_{01}$, $c_{10}$, $c_{11}$.
The functions $\IsFlipTree_1()$ and $\IsFlipTree_2()$ require storing only a constant number of ordered rooted trees with $n$ edges (=bitstrings of length $2n$) for lexicographic comparisons, so they also require space $\cO(n)$. 
Furthermore, it is readily checked that the additional space required by each of the calling functions up to the top-level algorithm $\HamCycle()$ is only $\cO(n)$. 
This proves that the total space required by our algorithm is $\cO(n)$.

\section{Acknowledgements}

We thank Günter Rote for persistently raising the question whether the proof of Theorem~\ref{thm:middle-levels} could be turned into an efficient algorithm. 
We also thank the referees of this paper and of the extended abstract that appeared in the proceedings of the European Symposium on Algorithms (ESA) 2015 for numerous valuable comments that helped improving the presentation of this work and the C++ code.

\bibliographystyle{alpha}
\bibliography{refs}

\end{document}